\documentclass[journal]{IEEEtran}
\IEEEoverridecommandlockouts
\usepackage{dsfont}

\usepackage{cite}
\usepackage{amsmath,amssymb,amsfonts}
\usepackage{amsmath, amsthm}
\usepackage{algorithmic}
\usepackage{cases}
\usepackage{authblk}
\usepackage{enumitem}
\usepackage[linesnumbered,ruled,vlined]{algorithm2e}

\usepackage{graphicx}
\usepackage{amsmath}
\usepackage{textcomp}
\usepackage{caption}
\usepackage{subcaption}
\usepackage{xcolor}
\usepackage{upgreek}

\newtheorem{condition}{Condition}
\newtheorem{remark}{Remark}

\SetKwFunction{MyFunction}{MyFunction(arg1, arg2)}

\newtheorem{theorem}{Theorem}
\newtheorem{lemma}{Lemma}
\usepackage{bm}

\DeclareMathOperator*{\argmin}{arg\,min}
\newenvironment{talign}
 {\let\displaystyle\textstyle\align}
 {\endalign}

\def\BibTeX{{\rm B\kern-.05em{\sc i\kern-.025em b}\kern-.08em
    T\kern-.1667em\lower.7ex\hbox{E}\kern-.125emX}}

\begin{document}
\title{Resource Allocation for the Training of Image Semantic Communication Networks}
% {\footnotesize \textsuperscript{*}Note: Sub-titles are not captured in Xplore and
% should not be used}
% \thanks{Identify applicable funding agency here. If none, delete this.}
% }
% \author{Yang~Li, Xinyu~Zhou, Jun~Zhao
% \thanks{The authors are all with the School of Computer Science and Engineering, Nanyang Technological University, Singapore. Yang Li and Xinyu Zhou are also with ERI@N, Interdisciplinary Graduate Programme, Nanyang Technological University, Singapore. Contact: yang048@e.ntu.edu.sg, xinyu003@e.ntu.edu.sg, JunZHAO@ntu.edu.sg
% } 

% }
% \author[1]{Xinyu~Zhou}
% \author[2]{Jun Zhao}

% \affil[2]{\normalsize School of Computer Science \& Engineering, Nanyang Technological University, Singapore} 
% \affil[ ]{\normalsize 
% yang048@e.ntu.edu.sg,
% xinyu003@e.ntu.edu.sg,  junzhao@ntu.edu.sg} 

\author{\IEEEauthorblockN{Yang~Li, Xinyu~Zhou, Jun~Zhao} 
\thanks{Y. Li and X. Zhou are with the Graduate College at Nanyang Technological University (NTU), Singapore 639798. J. Zhao is with the College of Computing and Data Science (CCDS) at Nanyang Technological University (NTU), Singapore 639798. CCDS was formerly known as the School of Computer Science and Engineering (SCSE). Y. Li and X. Zhou are both NTU PhD students supervised by J. Zhao. (Emails:  yang048@e.ntu.edu.sg, xinyu003@e.ntu.edu.sg, JunZhao@ntu.edu.sg)\newline \indent Corresponding author: Jun Zhao. \newline \indent
This research is supported partly by Singapore Ministry of Education Academic Research Fund Tier 1 RT5/23, Tier 1 RG90/22, Nanyang Technological University (NTU)-Wallenberg AI, Autonomous Systems and Software Program (WASP) Joint Project, and  Imperial-Nanyang Technological University Collaboration Fund INCF-2024-008.
} 
}

\maketitle
\begin{abstract}

Semantic communication is a new paradigm that aims at providing more efficient communication for the next-generation wireless network. It focuses on transmitting extracted, meaningful information instead of the raw data.
However, deep learning-enabled image semantic communication models often require a significant amount of time and energy for training, which is unacceptable, especially for mobile devices. 
To solve this challenge, 
our paper first introduces a distributed image semantic communication system where the base station and local devices will collaboratively train the models for uplink communication. 
Furthermore, we formulate a joint optimization problem to balance time and energy consumption on the local devices during training while ensuring effective model performance. 
An adaptable resource allocation algorithm is proposed to meet requirements under different scenarios, and its time complexity, solution quality, and convergence are thoroughly analyzed. 
Experimental results demonstrate the superiority of our algorithm in resource allocation optimization against existing benchmarks and discuss its impact on the performance of image semantic communication systems.

\end{abstract}


\section{Introduction}
Current information technologies near the Shannon capacity limits may not meet the rising data demands of the 6G era \cite{wang2022transformer}. The upcoming 6G wireless technology is critical for emerging applications, e.g., Internet of Everything (IoE) \cite{khan2022digital,zhang2024joint}, virtual reality \cite{wang2023survey}, autonomous driving \cite{qu2024model,ye2021joint}, wireless sensing \cite{wang2019continuous}, and other application domains \cite{lin2023blockchain,park2023enabling,yang2019miras}.
To support these advances, 
this generation must offer high-speed, low-latency, and energy-efficient \cite{pala2023spectral,hu2023computation} communication within the constraints of limited spectrum and power resources.

\textit{Semantic Communication} (SemCom) \cite{shannon1949mathematical} has attracted a lot of attention recently as a promising solution to the limitations of current information technologies. Unlike traditional methods that send all the data, SemCom is task-oriented and enables devices only to transmit desired semantic information extracted from the original data (i.e., text, image, and speech), thereby enhancing communication efficiency \cite{qin2021semantic}. Compared to textual data, the semantic information in an image is inherently implicit and depends on contextual knowledge. Hence, Deep Learning (DL) technology has been widely used in the image SemCom due to its advantage of deciphering the meaning of a figure. Specifically, 
the transmitter could use the DL model to encode the image into a low-dimensional vector to achieve highly efficient compression, which is then restored by the corresponding decoder at the receiver \cite{zhang2022deep}. 
Furthermore, the learning capability of neural networks allows DL models to adjust to varying channel conditions, enhancing their ability to mitigate issues like noise interference.

Implementing DL-enabled image SemCom in a practical communication network also faces a shift from the current device-to-device manner to a sophisticated, network-centric system. Here, we consider a distributed image SemCom system, where each local device is equipped with a personalized DL-enabled SemCom model. During uplink communication, the local devices compress the image and transmit the semantic information. Upon reception, the base station decodes and recovers it. The system is continually trained to meet device requirements and changing channel conditions.

\textbf{Motivation and Challenges.}
However, training a distributed image SemCom system within a real communication network presents several challenges: (1) \textbf{Power issue}. Although training DL models on local devices has become practical with the development in computation capability, the power issue has been a major concern, especially for mobile devices (i.e., phones and laptops). The training of DL models often brings a large amount of energy consumption on local devices, which may be unacceptable for users. (2) \textbf{Communication burden}. In the SemCom process, since the encoder and decoder are located on separate devices, this requires communication between the two devices (e.g., semantic information, model parameters, gradients) to complete the training. Due to the local devices' limited communication resources (e.g., transmission power, bandwidth), this process tends to result in high latency and energy costs. (3) \textbf{Device variability}. Different devices in the network may have different computing capabilities and individual requirements for SemCom (e.g., information compression rate). How to take into account this variability while ensuring effective training for each device is also a tricky problem.

To solve these challenges, we propose a DL-enabled image SemCom distributed system and aim to minimize the overall consumption to achieve more efficient training. The main \textbf{contributions} of our paper are summarized as follows:
\begin{itemize}
\item We propose a DL-enabled image SemCom system, where the local devices collaborate with the base station to train real-time models, adapting to the varying channel conditions and optimizing model performance continuously.
\item A joint optimization problem is formulated that aims to minimize a weighted sum of latency and energy consumption on local devices. We also investigate the relationship between model performance and relevant variables to guarantee the performance of SemCom.

\item An adaptable resource allocation algorithm is developed. Owing to the problem's intricate nature, \mbox{non-convex} ratios, and other complexities, we first decouple the original problem into tractable subproblems. Then, a tailored combination of the Lagrange function and a Newton-like approach is utilized to solve these subproblems. This algorithm dynamically adjusts to different system requirements. We also comprehensively analyze the time complexity, solution quality, and convergence.

\item We implement extensive experiments to compare the proposed algorithm with benchmarks. Our results illustrate the superior performance of our algorithm in optimizing resource allocation while ensuring performance.
\end{itemize}

This paper is structured as follows: Section \ref{Sec:related Work} reviews related work. Section \ref{Sec:System model} introduces the system model and notations. Section \ref{Sec:Problem Formulation} discusses problem formulation, followed by a proposed resource allocation algorithm in Section \ref{Sec:Solution}. Section \ref{Sec:time Complexity} analyzes the algorithm's performance. Simulation results are presented in Section \ref{Sec:Simulation}, and the paper concludes with Section \ref{Sec:Conclusion}.

\section{Literature Review}\label{Sec:related Work}
This section surveys the literature on DL-enabled semantic communication techniques and resource allocation strategies in SemCom networks.

\subsection{DL-enabled Semantic Communication Model}

With the advent of powerful deep learning technologies, a number of studies have investigated the applications of DL-enabled semantic communication on text~\cite{xie2021deep,liu2022extended,kutay2023semantic,peng2022robust}, images~\cite{huang2021deep, lokumarambage2023wireless}, speech~\cite{grassucci2024diffusion}, and videos~\cite{jiang2022wireless}. Specifically, Xie~\emph{et~al.}~\cite{xie2021deep} proposed DeepSC, a SemCom system using deep
learning to maximize communication efficiency by minimizing semantic errors during text transmission.
It employed a Transformer~\cite{vaswani2017attention}-based architecture and further introduced transfer learning~\cite{weiss2016survey} to enhance adaptability across different communication environments. A new metric, sentence similarity, is also introduced to evaluate performance. Liu~\emph{et~al.}~\cite{liu2022extended} presented an Extended Context-based Semantic Communication (ECSC) system, enhancing text transmission by fully leveraging context information both within and between sentences. 
The ECSC encoder captures context with attention mechanisms, while the decoder uses Transformer-XL~\cite{dai2019transformer} for enhanced semantic recovery. Kutay~\emph{et~al.}~\cite{kutay2023semantic} utilized SBERT~\cite{reimers2019sentence} to obtain sentence embeddings and a semantic distortion metric, preserving text meaning while achieving high compression results. They further integrate semantic quantization with semantic clustering, generalizing well across diverse text classification datasets. Peng~\emph{et~al.}~\cite{peng2022robust} introduced R-DeepSC, a robust deep learning-based system using calibrated self-attention and adversarial training to mitigate the effects of ``semantic noise'' in text transmission. R-DeepSC outperforms traditional models only considering physical noise, demonstrating strong resilience to various signal-to-noise ratios.

Other than text, DL-enabled semantic communication also demonstrates satisfying performance on more complex data types (e.g., images and videos). For image transmission, Huang~\emph{et~al.}~\cite{huang2021deep} presented a Generative Adversarial Networks (GANs)-based image SemCom framework, achieving efficient semantic reconstructions.
Lokumarambage~\emph{et~al.}~\cite{lokumarambage2023wireless} proposed an end-to-end image SemCom system, using a pre-trained GAN at the receiver to reconstruct realistic images from semantic segmentation maps generated by the transmitter. In particular, the deployment of GAN allows it to combat noise in poor channel conditions, outperforming conventional compression methods.
Grassucci~\emph{et~al.}~\cite{grassucci2024diffusion} proposed a generative audio semantic communication framework, employing a diffusion model to restore the received information from various degradations such as noise and missing parts. This approach focuses on semantic content rather than exact bitstream recovery, proving robustness to different channel conditions.
For video transmission, Jiang~\emph{et~al.}~\cite{jiang2022wireless} introduced a semantic video conferencing (SVC) network where only critical frames of videos are transmitted, focusing on motions due to static backgrounds and infrequent speaker changes.
Furthermore, they proposed an IR-HARQ framework with a semantic error detector and an SVC-CSI for channel feedback, improving error detection and transmission efficiency.

% Huang~\emph{et~al.}~\cite{huang2022toward} have introduced a Reinforcement Learning-based Adaptive Semantic Coding (RL-ASC) model, which extracts semantic features using a convolutional encoder and reconstructs images with a Generative Adversarial Network (GAN)-based decoder. A transformer-based multimodal SemCom system was proposed in \cite{xie2022task}, which simultaneously handles text and image input. 
% The work \cite{lokumarambage2023wireless} presented an image transmission system by transmitting semantic segmentation maps and using a pre-trained GAN for reconstruction. A ResNet-based semantic signal processing framework was proposed in \cite{kalfa2021towards}, which can be tailored for specific tasks efficiently through a goal filtering method.

% DL has shown its versatility in training diverse models for transmitting different data types, such as text \cite{xie2021deep}, images \cite{liu2021semantics,hu2022robust}, speech \cite{weng2021semantic}, and video \cite{jiang2022wireless}. These studies focused on training neural network-based encoders and decoders for SemCom but typically \textcolor{blue}{neglected} the multi-device model training within a network, an area our research addresses.

\subsection{Resource Allocation in SemCom Networks}

Most SemCom frameworks~\cite{xie2021deep,liu2022extended,peng2022robust,kutay2023semantic,huang2021deep, lokumarambage2023wireless,grassucci2024diffusion,jiang2022wireless,weng2021semantic} primarily focused on one-to-one device communication, and they often lack the adaptability required for large-scale scenarios. 
To bridge this gap, a number of recent studies~\cite{yan2022resource,yan2022qoe,wang2022performance, yang2023energy,zhang2023drl,zhang2023optimization} have shifted their focus towards developing both efficient and effective SemCom systems that are suitable for large-scale communication networks. 
Particularly, resource allocation acts as a significant factor in the above studies, as it is often crucial for improving the scalability and efficiency of networks, especially under resource-limited scenarios. 

Those studies~\cite{yan2022resource,yan2022qoe,wang2022performance, yang2023energy,zhang2023drl,zhang2023optimization} concentrated on optimizing network resource allocation to enhance metrics such as latency, energy consumption and signal quality. 
Specifically, Yan~\emph{et~al.}~\cite{yan2022resource} introduced the concept of semantic spectral efficiency (S-SE), a novel metric designed to optimize resource allocation through strategic channel assignments and semantic symbol transmission. A semantic-aware resource allocation algorithm is also proposed to maximize the overall S-SE of all users. 
Afterward, they extended their work in a later study~\cite{yan2022qoe}, developing an approximate measure of semantic entropy with a new quality-of-experience (QoE) model based on the measure.
Resources, including symbol quantity, channel assignment and power allocation, are optimized to maximize the overall QoE. 
Wang~\emph{et~al.}~\cite{wang2022performance} presented a metric of semantic similarity (MSS), which can jointly capture the semantic accuracy and completeness of the recovered information. 
They integrated a reinforcement learning (RL) algorithm and an attention-based network to optimize the resource allocation policy, maximizing the MSS of the whole network. Yang~\emph{et~al.}~\cite{yang2023energy} introduced a downlink SemCom network, minimizing overall energy usage by optimizing computation capacity, power control, rate allocation, etc. In addition, the SemCom model performance and latency are also considered constraints to the optimization problem. 
Zhang~\emph{et~al.}~\cite{zhang2023drl} considered a task-oriented downlink SemCom network and aimed to enhance the overall transmission efficiency. A deep deterministic policy gradient (DDPG) agent is developed to optimize the balance between data delivery and task accuracy by managing the semantic compression ratio, transmit power, and bandwidth allocation.
Zhang~\emph{et~al.}~\cite{zhang2023optimization} discussed an image SemCom network and aimed to minimize the average transmission latency of all users while satisfying the semantic reliability requirement. They managed to optimize resources such as user association, resource block (RB) allocation, and the selection of semantic information for transmission.

\subsection{Comparison between Our Work and Related Studies} 

Previous related works~\cite{yan2022resource,yan2022qoe,wang2022performance,yang2023energy,zhang2023drl,zhang2023optimization} investigating resource allocation in SemCom networks mainly focused on the inference period rather than model training period in our study. There are also a few studies~\cite{Nguyen2024FedSem,sun2024FedCMTSem,wang2024FedCL} discussing the training framework of SemCom networks, but they did not design the corresponding resource allocation strategy to optimize the system holistically.
In synthesizing the above discussion, we highlight the unique contribution of our work: jointly optimizing latency, energy consumption and model performance while training an image SemCom network.  
We note that while several advancements exist in each domain, our integrated approach presents an innovative method that has yet to be fully explored in the existing literature.

\begin{figure*}
    \centering
    \includegraphics[width =.95\textwidth]{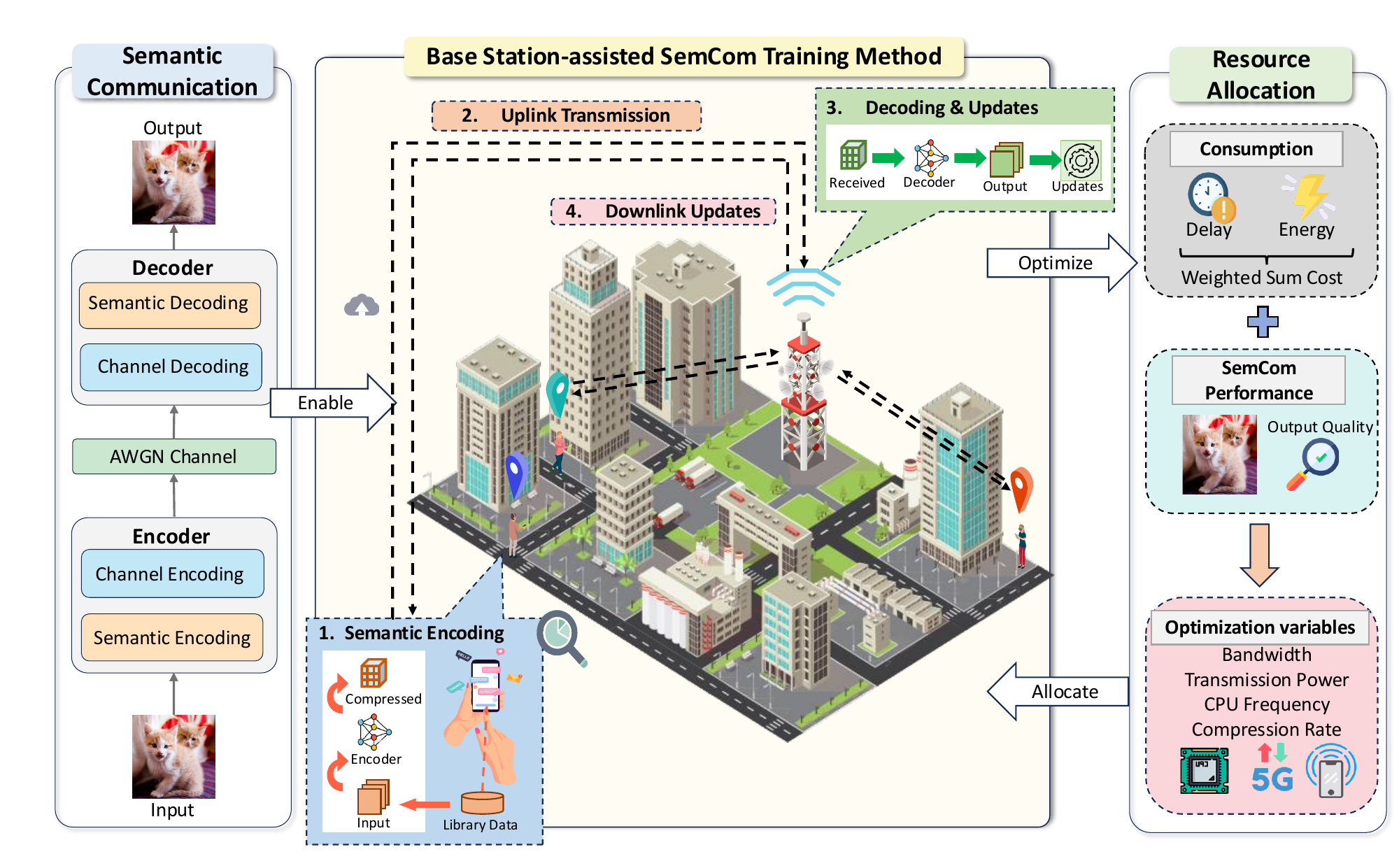}
    % \vspace{-5pt}
    \caption{System model.}
    % \vspace{-15pt}
    \label{fig:system_model}
\end{figure*}

\section{System Model}\label{Sec:System model}

In this paper, we consider an image SemCom network consisting of one Base Station (BS) and a set of mobile devices: $\mathcal{N}=\{1,2,\ldots,N\}$. 
Due to the limited communication resources of mobile devices, they are required to perform SemCom during uplink image transmission.
Given the varying channel conditions, the devices want to efficiently train their models for better SemCom performance.
Hence, as illustrated in Fig.~\ref{fig:system_model}, we have developed a training system that allows mobile devices and the BS to jointly train their SemCom models. 
We first elaborate on our training paradigm in Section~\ref{subsection:SemCom System}, and then we discuss the consumption model and SemCom performance in Sections~\ref{subsection:E,T} and~\ref{subsection:accuracy}, respectively.

\subsection{Training of Semantic Communication Network}\label{subsection:SemCom System}
\textbf{Semantic communication.}
In this paper, we utilize a DL-enabled model to achieve image SemCom. The main components are an encoder for image compression and a decoder for image reconstruction, as illustrated in the left side of Fig.~\ref{fig:system_model}.
% The input to the model is an original image, and the output will be the correspondingly reconstructed image after efficient transmission in channel.
Next, we detail the components and the model objective.

Given an input image $x$, the encoder performs semantic encoding to extract relevant image features. Then, it further transforms these features into a sequence of encoded symbols suitable for transmission, denoted as $X$. The whole encoder function could be expressed as follows:
% We use $f(\cdot\mid\theta)$ to denote the encoder functionis parameterized using a CNN with
% parameters θ.
\begin{align}
    X = f(x\mid\theta),
\end{align}
where $f(\cdot\mid\theta)$ represents the overall encoder function using a deep learning model parameterized by $\theta$.

The encoded sequence $X$ is transmitted through an AWGN channel, where it may encounter various disturbances. The corresponding output signal $Y$ could be modeled as:
\begin{align}
    Y = X+\mathcal{N}^{A},
\end{align}
where $\mathcal{N}^{A}$ denotes the AWGN noise, which is a Gaussian random variable with zero mean and variance $\sigma^2$.

After receiving signal $Y$, the encoder correspondingly performs channel decoding and semantic decoding, aiming to reconstruct the original image while minimizing distortion:
\begin{align}
    \hat{x}=f^{-1}(Y\mid\psi),
\end{align}
where $\hat{x}$ is the reconstructed image and $f^{-1}(\cdot\mid\psi)$ denotes the decoder function with a model parameterized by $\psi$.

The objective of both the encoder and decoder is to jointly
minimize the average distortion between the original 
image $x$ and its reconstruction $\hat{x}$. 
% For ease of representation, we use $\Theta$ to denote model parameters $(\theta,\overline{\theta})$:
\begin{align}
    (\theta^{*},\psi^{*}) = \argmin_{\theta,\psi}~\mathbb{E}[d(x,\hat{x})],
\end{align}
where $d(x,\hat{x})$ is the distortion measurement\footnote{Common measurements for image reconstruction include peak-signal-to-noise ratio (PSNR) and structural similarity index measure (SSIM), etc.} of the $x$ and $\hat{x}$. 

In our considered scenario, the semantic communication models are deployed at both mobile devices and the BS.
The devices utilize encoders to distill semantic information from images for uplink transmission. Upon reception, the BS reconstructs the images using the corresponding decoder.
Specifically, we employ a deep joint source and channel coding (JSCC) model~\cite{bourtsoulatze2019deep} as our SemCom model. The model aims to minimize the average mean squared error (MSE) loss between the original and reconstructed images. We defer a more detailed introduction of the model to Section~\ref{subsection:SemCom System} and describe below how to adapt it for training in our semantic communication network.

% The autoencoder could be regarded as 
% In particular, the encoder first implements semantic encoding and channel encoding, to convert the original input image into a sequence of encoded symbols for transmission.
% % The encoder function $f(\cdot\mid\theta)$ is parameterized using a DL model with parameters $\theta$.
% The encoded sequence undergoes channel-induced disturbances, modeled by non-trainable layers representing AWGN channels.
% The received signal is then decoded by the decoder, which aims to reconstruct the original image while minimizing distortion. 

% This end-to-end system is optimized through gradient computation and backpropagation to improve the fidelity and robustness of transmitted data against channel impairments.

% The architecture of the encoder includes two components: semantic encoding and channel encoding. The semantic encoding model, represented by $f_{\theta_1}(\cdot)$, is used to extract the critical semantic information from the original image. Concurrently, the channel encoder, denoted as $c_{\theta_2}(\cdot)$, is responsible for symbol generation, converting the extracted semantic information into symbols that facilitate transmission.

% Moving to the decoder’s end, it comprises channel decoding, $c_{\theta_3}^{-1}(\cdot)$, used for symbol detection, paired with semantic decoding, $f_{\theta_4}^{-1}(\cdot)$, dedicated to the task of image reconstruction. Then, a simple classifier is adopted to implement image classification to evaluate the reconstructed images' quality further.

\textbf{Iterative Training Method.} 
We propose an iterative training method for our SemCom network, detailed in Algorithm~\ref{Algorithm:Training}. 
Following~\cite{xie2021deep,van2022image,yang2024semantic}, we assume that both the mobile devices and the BS have access to a commonly shared training database. For simplicity, we denote the training data for a specific device $n$ as $Z_n$, assuming that the data has been properly aligned between the devices and the BS beforehand.
The process begins by setting a maximum number of iterations $L$ and initializing the iteration counter $l=0$. Each iteration $l$ of the training process includes the following steps:
\begin{itemize}
    \item Encoding and Transmission:
    \begin{itemize}
        \item Each device $n$ encodes its data $Z_n$ into a semantic signal: $X_n=f(Z_n\mid\theta_n^l)$.
        \item The signal $X_n$ is then transmitted to the BS.
    \end{itemize}
    \item Reception and Decoding at BS:
    \begin{itemize}
        \item The BS receives signals $\bm{Y} = [Y_1, \dots, Y_N]$, where each $Y_n = X_n + \mathcal{N}^{A}_n$.
        \item The BS decodes signals $\bm{Y}$ using the decoder to obtain: $\boldsymbol{\widehat{Z}}= [\widehat{Z}_1,\ldots,\widehat{Z}_N],$ where $\widehat{Z}_n=
        f^{-1}(Y_n \mid \psi_n^l)$.
    \end{itemize}
    \item Loss Computation and Model Update:
    \begin{itemize}
        \item The BS computes the MSE loss $\mathcal{L}_{\textnormal{MSE}}(Z_n, \widehat{Z}_n)$ to evaluate model performance of each device $n$.
        \item The BS updates the whole model to $(\theta_n^{l+1},\psi_n^{l+1})$ by applying the gradient descent on loss.
        % $\boldsymbol{\Theta}^{l+1} \leftarrow \boldsymbol{\Theta}^{l} - \eta \cdot \nabla_{\boldsymbol{\Theta}^{l}}\mathcal{L}_{\textnormal{MSE}}(\boldsymbol{Z}, \boldsymbol{\widehat{Z}})$.
    \end{itemize}
    \item Distribution of Updated Parameters:
    \begin{itemize}
        \item The updated encoder $\theta_n^{l+1}$ is distributed back to each device $n$ for local update.
    \end{itemize}
\end{itemize}

Generally, Algorithm~\ref{Algorithm:Training} is designed to continuously adapt the SemCom models to improve performance and robustness, accounting for channel conditions and device-specific data.

\begin{algorithm}
\footnotesize
\caption{Iterative Training Method} 
\label{Algorithm:Training}
\KwIn{Models $[(\theta_n^{0},\psi_n^{0})\mid_{n \in \mathcal{N}}]$; 
\newline
Training data $\boldsymbol{Z}=[Z_n|_{n \in \mathcal{N}}]$;
% \newline
% Initial model parameters $\boldsymbol{\theta}^{(0)}=[\theta_n^{(0)}|_{n \in \mathcal{N}}]$, 
}

\KwOut{Trained models $[(\theta_n^{l},\psi_n^{l})\mid_{n \in \mathcal{N}}]$.}

Set maximum iterations $L$ and iteration counter $l=0$;

\While{$l \leq L$ or not convergence}{
\For{each mobile device $n$ \underline{in parallel}}{

$X_n \leftarrow f(Z_n\mid\theta_n^{l})$; 

Transmit $X_n$ to the base station; 
}

\For{BS to each device $n$ \underline{in parallel}}{
Receive signal $Y_n$, where $Y_n = X_n + \mathcal{N}^{A}_n$;

$\widehat{Z}_n \leftarrow f^{-1}(Y_n\mid \psi_n^{l})$; 

Compute the MSE loss: $\mathcal{L}_{\textnormal{MSE}}(Z_n,\widehat{Z}_n)$;

Update the whole model: 
$(\theta_n^{l+1},\psi_n^{l+1}) \leftarrow (\theta_n^{l},\psi_n^{l}) - \eta \cdot \nabla_{\theta_n^{l},\psi_n^{l}}\mathcal{L}_{\textnormal{MSE}}(Z_n,\widehat{Z}_n)$; 

Distribute $\theta_n^{l+1}$ to each device $n$; 
}
Let $l \leftarrow l+1$; 
}
\end{algorithm}

% \begin{figure*}
%     \centering
%     \includegraphics[width =\textwidth]{figs/FL_Model.jpg}
%     % \vspace{-5pt}
%     \caption{System model.}
%     % \vspace{-15pt}
%     \label{fig:FL}
% \end{figure*}

% The system design was inspired by the JSAC paper \cite{zhang2022deep}, as shown in Fig.\ref{fig:2}.

% \begin{figure*}
%     \centering
%     \includegraphics[width =0.8\linewidth]{figs/JSAC.png}
%     % \vspace{-5pt}
%     \caption{JSAC.}
%     % \vspace{-15pt}
%     \label{fig:2}
% \end{figure*}

\begin{table}[!t]
\centering
\caption{Notations.}
\begin{tabular}{l|l}
\hline
Description   & Symbol    \\ \hline
The number of total devices   & $N$  \\
Original data on device $n$ & $Z_n$ \\
Reconstructed data at the base station& $\widehat{Z}_n$ \\
Encoder and decoder of device $n$ & $\theta_{n}, \psi_{n}$ \\
Encoded signal of device $n$ & $X_n$ \\
Channel output of device $n$ & $Y_n$ \\
Total available bandwidth & $B_{\textnormal{total}}$ \\
Allocated bandwidth to device $n$ & $B_n$ \\
The power density of Gaussian noise & $N_0$ \\
The number of CPU cycles per sample on device $n$ & $c_{1n}$ \\
The number of CPU cycles per sample on the BS & $c_{2n}$ \\
The number of samples on device $n$ & $D_n$ \\
Local computation frequency of device $n$ & $f_n$ \\
Allocated frequency to device $n$ on the BS & $h_n$ \\
Transmission power of device $n$ & $p_n$ \\
Uplink transmission rate & $r_n$ \\
Signal-to-Noise Ratio & $S_n$ \\
Local computation time on device $n$ & $T_n^{\text{cmp}}$ \\
Uplink transmission time on device $n$ & $T_n^{\text{up}}$ \\
The effective switched capacitance  & $\kappa$\\
% Ground truth label  &  $K^n$ \\
% Retrieved label  &  $\widehat{K}^n$ \\
Size of per sample data on device $n$  & $d_{n}$ \\
Compression rate of semantic communication & $\rho_n$ \\
Local computation energy of device $n$ & $E_n^{\text{cmp}}$ \\
Uplink transmission energy of device $n$ & $E_n^{\text{up}}$ \\
Overall energy consumption among devices & $E$ \\
Weight parameters for time and energy & $\omega_1,\omega_2$\\
Maximum transmission power & $p_n^{\textnormal{max}}$\\
Maximum frequency on device $n$ & $f_n^{\textnormal{max}}$\\
Maximum allocated frequency to device $n$ on the BS & $h_n^{\textnormal{max}}$\\
Maximum transmission power & $p_n^{\textnormal{max}}$\\
Maximum and minimum compression rate & 
$\rho_n^{\textnormal{max}},\rho_n^{\min}$\\
Minimum PSNR requirement of device $n$ & $P_n^{\textnormal{min}}$ \\
% Minimum transmission rate & $\widehat{r}_n(f_n,\rho_n,\hspace{.8pt}\mathcal{T})$ \\
% Auxiliary Variables & $\mathcal{T},\bm{\gamma},\bm{\delta}$ \\
% Lagrange multipliers & $\bm{\alpha},\bm{\beta},\zeta,\bm{\eta},\bm{}\nu,\bm{\iota}$ \\
\hline
\end{tabular}
\label{tab:notations}
\end{table}

% To achieve this objective, we will formulate a joint optimization problem that addresses both time and energy consumption while also taking accuracy into account. In the following subsections \ref{subsection:E,T} and \ref{subsection:accuracy}, we will discuss time completion, energy consumption and model accuracy, respectively.

\subsection{Computation and Transmission Model}\label{subsection:E,T}

In this subsection, we give a detailed introduction to our computation and transmission model within the SemCom system. Considering the BS usually possesses abundant communication resources (e.g., high transmission power and abundant bandwidth), our primary focus lies in minimizing the consumption of mobile devices and the computation consumption of the BS~\cite{luo2020hfel,yang2020energy,dinh2020federated,yang2023detfed}.
We first introduce the frequency division multiplexing technique used, followed by the formulation of time completion as well as energy consumption.

\textbf{FDMA.} Following \cite{luo2020hfel,yang2020energy,zhou2022joint}, we adopt FDMA technology, which could be implemented simply for mobile devices with limited computation capability \cite{li2021talk}. Specifically, given a total bandwidth $B_{\textnormal{total}}$, the bandwidth allocated to each user $n$ is denoted as $B_n$, and $\sum_{n=1}^{N} B_n \leq B_{\textnormal{total}}$.

\textbf{Time consumption of mobile devices.}
In each training iteration, mobile device $n$ encodes the original image data and then transmits the encoded signal to the BS.
Hence, we first analyze the computation time consumption for encoding.
Assuming that the CPU frequency $f_n$ does not change over time until it has been optimized, the computation time on device $n$ is:
\begin{align}
    T_n^{\text{cmp}} =  \frac{c_{1n} D_n}{f_n},\label{equa:T_cmp_D}
\end{align}
where $c_{1n}$ and $D_n$ are the required CPU cycles per sample and the number of samples on device $n$, respectively. 
Based on the Shannon formula, we define the uplink transmission rate of device $n$ as follows:
\begin{align}
    r_n= B_n \log_2{\big(1+\frac{p_ng_n}{N_0B_n}\big)},
\end{align}
where $p_n$ is the transmission power, $g_n$ is the channel gain and $N_0$ is the power spectral density of Gaussian noise. 
Let $d_{n}$ represent the size of per sample on device $n$, with a semantic information compression rate of $\rho_n$. According to Algorithm~1, the uplink transmission time $T_n^{\text{up}}$ of device $n$ is:
\begin{align}\label{transmission time}
    T_n^{\text{up}} =  \frac{\rho_nd_nD_n}{r_n},
\end{align}
where $\rho_n$ is the compression rate and is defined as the ratio of the size of semantic information to that of the original data. Thereafter, the overall time consumption $T_n^{D}$ on device $n$ in one iteration is expressed as:
\begin{align}
    T_n^{D} = T_n^{\text{cmp}}+T_n^{\text{up}}.\label{equa:T_up}
\end{align}

\textbf{Time consumption of the BS.} Upon receiving the encoded signal, the BS decodes it for reconstruction and updates the SemCom model. The updated model is finally sent back to each device $n$. Since computation consumption is the main source of overhead, we define the BS's time corresponding to each device $n$ as follows: 
\begin{align}
    T_n^{S} =  \frac{c_{2n} D_n}{h_n},\label{equa:T_cmp_BS}
\end{align}
where $c_{2n}$ denotes the required CPU cycles for decoding and model updating per sample, and $h_n$ is the assigned frequency to device $n$ on the BS.

\textbf{Energy consumption of mobile devices.} 
Following \cite{dinh2020federated}, we define the computation energy consumption of device $n$ during one round of training as follows:
\begin{align}
    E_n^{\text{cmp}} =  \kappa c_{1n} D_n f_n^2,
\end{align}
where $\kappa$ denotes the effective capacitance coefficient. Since the transmission time $T_n^{\text{up}}$ is already obtained in (\ref{transmission time}), the transmission energy consumption $E_n^{\text{up}}$ could be formulated as: 
\begin{align}
    E_n^{\text{up}} = p_n T_n^{\text{up}} = p_n \frac{ \rho_nd_{n}D_n}{r_n}.
\end{align}
Hence, the overall energy consumption of device $n$ is:
\begin{align}
    E_n^{D}= E_n^{\text{cmp}}+E_n^{\text{up}}.
\end{align}

\textbf{Energy consumption of the BS.} Based on computation time $T_n^{S}$ in (\ref{equa:T_cmp_BS}), we define $E_n^{S}$ as the BS's energy consumption corresponding to each device $n$:
\begin{align}
    E_n^{S} =  \kappa c_{2n} D_n h_n^2.
\end{align}

% Until now, we can generalize the overall consumption to each device $n$ within the network. In one iteration, 
% the overall time consumption $T_n$ is:
% \begin{align}
%     T_n = T_n^{D}+T_n^{S}=
% \end{align}

\subsection{SemCom Performance Analysis }\label{subsection:accuracy}
We employ the Peak Signal-to-Noise Ratio (PSNR) metric to evaluate the performance of the image SemCom model.
PSNR is computed from the MSE loss and indicates the similarity between the original and reconstructed images at the pixel level~\cite{huang2021deep,huang2022toward,wu2023semantic,dong2022innovative}.
In our scenario, given the SemCom models, we examine PSNR's dependence on two key factors: the semantic information compression rate $\rho_n$ and the Signal-to-noise ratio (SNR) $S_n$, where $S_n = \frac{p_ng_n}{N_0B_n}$. Higher values of $\rho_n$ and $S_n$ improve the quality of the reconstructed image, as they mitigate information loss caused by compression and enhance transmission reliability.
Based on the observed characteristics of PSNR in SemCom~\cite{bourtsoulatze2019deep} and other practical applications~\cite{kumari2019fair, xiong2020reward}, the PSNR value can be modeled as a concave function $P(\cdot)$ concerning $\rho_n$ and $S_n$. 
Hence, we make an assumption of $P(\rho_n,S_n)$ that is presented in Condition~\ref{condition:accuracy_function}. 

\begin{condition}\label{condition:accuracy_function}
    The PSNR function $P(\rho_n,S_n)$ for any $n \in \mathcal{N}$ is concave and non-decreasing with respect to $\rho_n > 0$ and $S_n > 0$; i.e., $P^{\prime\prime}(x) < 0$ and $P^{\prime}(x) > 0$ for $x > 0$, where $x$ denotes $\rho_n$ or $S_n$.
\end{condition}

The concavity of $P(\rho_n,S_n)$ means diminishing marginal gain.
For the specific expression of $P(x)$ used in the experiments, we will discuss it in Section VII-B.

\section{Joint Optimization Problem Formulation}\label{Sec:Problem Formulation}

In this section, we begin with formulating the joint optimization problem, followed by a discussion on the challenges and strategies for solving it.

\subsection{Problem Formulation}

The training of the SemCom network can cause considerable overhead to mobile devices and the base station.
Due to limited resources, one of our goals is to minimize the overall time consumption in one training iteration:
\begin{align}\label{obj:1}
% \min_{\bm{f},\bm{g},\bm{p},\bm{B},\bm{\lambda}} 
\min \left\{\max_{n\in\mathcal{N}}\{T_n\}\right\},
\end{align}
where $T_n=T_n^{D}+T_n^{S}$ denotes device $n$'s  total training time consumption. The $\max$ operation in (\ref{obj:1}) is for fairness (i.e., devices that become ``bottlenecks'' in time consumption will be continuously optimized). 
% By doing so, we can use limited resources to fairly minimize all devices' training time consumption.
Concurrently, we also aim to minimize the overall energy consumption within the network:
\begin{align}\label{obj:2}
    \min \left\{\sum_{n=1}^NE_n \right\},
\end{align}
where $E_n = E_n^{D} + E_n^{S}$ denotes each device $n$'s total energy consumption. Synthesizing the above two objectives (\ref{obj:1}) and (\ref{obj:2}), a joint optimization problem can be formulated as follows:
\begin{subequations}\label{Original_Problem}
\begin{align}
    \mathbb{P}_1:~
     \min_{\boldsymbol{p},\boldsymbol{B},\boldsymbol{f},\boldsymbol{h},\boldsymbol{\rho}} &
     \left\{ \omega_1 \max_{n\in\mathcal{N}}\{T_n\}
     +\omega_2 \sum_{n=1}^N E_n \right\}
     ,
    \tag{\ref{Original_Problem}} \\
    \text{s.t.}:~ 
    % & \rho_n^{\min} \leq \rho_n \leq \rho_n^{\textnormal{max}} \\
    & p_n \leq p_n^{\textnormal{max}},~\forall n \in \mathcal{N}, \label{constra:p}\\
    & f_n \leq f_n^{\textnormal{max}},~\forall n \in \mathcal{N}, \label{constra:f}\\
    % & B_n^{\min} \leq B_n,~\forall n \in \mathcal{N}, \label{constra:B} \\
    & \sum_{n=1}^{N} B_n \leq B_{\textnormal{total}}, \label{constra:B_sum} \\
    & h_n \leq h_n^{\textnormal{max}},~\forall n \in \mathcal{N}\label{constra:h_sum} \\
    & P_n^{\textnormal{min}} \leq P_n(\rho_n,S_n(p_n,B_n)),~\forall n \in \mathcal{N}, \label{constra:A}\\
    & \rho_n^{\min} \leq \rho_n \leq \rho_n^{\textnormal{max}},~\forall n \in \mathcal{N}, \label{constra:rho}
\end{align}
\end{subequations}
where $\omega_1, \omega_2 \in [0,1]$ are weights for time and energy consumption, respectively, and we can enforce $\omega_1+\omega_2=1$ after normalization. Constraints~(\ref{constra:p}) and (\ref{constra:f}) set the ranges of transmission power and computation frequency for each device $n$. 
Constraint~(\ref{constra:B_sum}) imposes an upper limit on the total available uplink bandwidth. 
Constraint~(\ref{constra:h_sum}) sets the maximum available computation resources for each device $n$ at the BS.
$P_n^{\textnormal{min}}$ in constraint~(\ref{constra:A}) is the required minimum model accuracy for device $n$, and constraint~(\ref{constra:rho}) sets the range of values for the compression rate.

\subsection{Challenges and Strategies of Solving Problem $\mathbb{P}_1$ }

Problem $\mathbb{P}_1$ has a complex construction form and is subject to complicated constraints (\ref{constra:p})-(\ref{constra:rho}). In the objective function, both $T_n^{\text{up}}$ and $ E_n^{\text{up}}$ are non-convex and compound $p_n,B_n,\rho_n$ three variables, making them extremely complex. The presence of the max function further complicates the problem-solving process. 
Besides, the joint concavity of  $A_n(\rho_n,S_n)$ in constraint (\ref{constra:A}) with respect to $\rho_n$ and $S_n$ is also not guaranteed, as we discuss in Section~\ref{subsection:accuracy}.

To streamline the problem $\mathbb{P}_1$ for efficient resolution, we first introduce an auxiliary variable $\mathcal{T}$. This variable replaces the time completion $T$ in (\ref{Original_Problem}), accompanied by an additional constraint. Consequently, we reformulate $\mathbb{P}_1$ as follows:
% Next, we discuss how to transform the original problem $\mathbb{P}_1$ into a more simplified form for efficient resolution.
% Firstly, we introduce an auxiliary variable $\mathcal{T}$ to replace the original time completion $T$ in  (\ref{Original_Problem}) with an additional constraint. Then we rewrite $\mathbb{P}_1$ as: 
\begin{subequations}\label{Problem2}
\begin{align}
    \mathbb{P}_2:~
    \min_{\boldsymbol{p},\boldsymbol{B},\boldsymbol{f},\boldsymbol{\rho}, \mathcal{T}}~ 
    & \left\{\omega_1 \mathcal{T} +\omega_2 \sum_{n=1}^N E_n \right\},
    \tag{\ref{Problem2}} \\
    \text{s.t.}:~ 
    &\text{(\ref{constra:p}), (\ref{constra:f}), (\ref{constra:B_sum}), 
    (\ref{constra:h_sum}), 
    (\ref{constra:A}), (\ref{constra:rho}),} \nonumber\\
    &  T_n^{\text{up}}+ T_n^{\text{cmp}}+ T_n^{S} \leq \mathcal{T}, ~\forall n \in \mathcal{N}. \label{constra:Told}
\end{align}
\end{subequations}

However, the problem of $T_n^{\text{up}}$ and $ E_n^{\text{up}}$ being non-convex and $A_n(\rho_n,S_n)$ being non-concave remains unsolved. Considering that they are coupled with several variables, we decompose $\mathbb{P}_2$ into two subproblems for further simplification. Since $A_n(\rho_n,S_n)$ is not guaranteed to be jointly concave with respect to $\rho_n$ and $S_n$, and variables $p_n$ and $B_n$ appear together in the term $S_n = \frac{p_ng_n}{N_0B_n}$, $\mathbb{P}_3$ involves optimization variables $f_n$, $\rho_n$ and $\mathcal{T}$, while $\mathbb{P}_4$ involves variables $p_n$ and $B_n$. 
\begin{subequations}\label{Subproblem1}
\begin{align}
\mathbb{P}_3(\boldsymbol{p},\boldsymbol{B}):~ 
    \min_{\boldsymbol{f},\boldsymbol{\rho},\boldsymbol{h}, \mathcal{T}}~ 
    & \left\{\omega_1 \mathcal{T} +\omega_2 \sum_{n=1}^N (E_n^{\textnormal{cmp}}+E_n^{S})\right\}
    % \kappa c_n D_n f_n^2
    ,
    \tag{\ref{Subproblem1}} \\
    \text{s.t.}:~ 
    & \text{(\ref{constra:f}), 
    (\ref{constra:h_sum}), (\ref{constra:A}),  (\ref{constra:rho}), (\ref{constra:Told})}. \nonumber
\end{align}
\end{subequations}
\begin{subequations}\label{Subproblem2}
\begin{align}\mathbb{P}_4(\boldsymbol{f},\boldsymbol{\rho}, \mathcal{T}\hspace{.8pt}):~
    \min_{\boldsymbol{p},\boldsymbol{B}}~ 
    & \left\{\omega_2 \sum_{n=1}^{N} E_n^{\text{up}}  \right\}
    % p_n \frac{ \rho_n d_{n}}{r_n}
    ,
    \tag{\ref{Subproblem2}} \\
    \text{s.t.}:~ 
    & \text{(\ref{constra:p}), 
    %(\ref{constra:B}),
    (\ref{constra:B_sum}),   (\ref{constra:A}),
    (\ref{constra:Told})}. \nonumber
\end{align}
\end{subequations}
% To solve problem $\mathbb{P}_1$, we adopt the following
% alternating optimization process:
% \begin{itemize}
% \item given $\boldsymbol{p} $ and $\boldsymbol{B}$, we solve $\mathbb{P}_3$ to optimize $\boldsymbol{f}$, $\boldsymbol{\rho}$, and $\mathcal{T}$
% \item given $\boldsymbol{f}$, $\boldsymbol{\rho}$, and $\mathcal{T}$, we solve $\mathbb{P}_4$ to optimize $\boldsymbol{p} $ and $\boldsymbol{B}$
% \end{itemize}
To solve problem $\mathbb{P}_1$, we adopt the following
alternating optimization process:
\begin{itemize}
\item given $\boldsymbol{p} $ and $\boldsymbol{B}$, we solve $\mathbb{P}_3$ to optimize $\boldsymbol{f}$, $\boldsymbol{h}$, $\boldsymbol{\rho}$, and $\mathcal{T}.$
\item given $\boldsymbol{f}$, $\boldsymbol{h}$, $\boldsymbol{\rho}$,  and $\mathcal{T}$, we solve $\mathbb{P}_4$ to optimize $\boldsymbol{p} $ and $\boldsymbol{B}.$
\end{itemize}
It is worth noting that despite the decomposition, $\mathbb{P}_4$ is still a challenging pseudoconvex ratio as validated by Lemma \ref{lemma1}:

\begin{lemma}\label{lemma1}
    For Problem $\mathbb{P}_4$, we have:
    \begin{itemize}
        \item Rate $r_n$ is jointly concave with respect to $p_n$ and $B_n$;
        \item Ratio $\frac{p_n\rho_nd_{n}}{r_n}$ in $E_n^{\text{up}}$ is jointly pseudoconvex with regard to  $p_n$ and $B_n$.
    \end{itemize}
\end{lemma}

\begin{proof}
    For the rate $r_n$ as a function of $p_n$ and $B_n$, we calculate the Hessian matrix as follows:
    \begin{align}
    H = 
    \begin{bmatrix}
    -\frac{g_n^2}{B_nN_0^2{(1+\frac{p_ng_n}{N_0B_n})}^2\ln{2}} & 
    \frac{p_n g_n^2}{B_n^2N_0^2{(1+\frac{p_ng_n}{N_0B_n})}^2\ln{2}} \\
    \frac{p_n g_n^2}{B_n^2N_0^2{(1+\frac{p_ng_n}{N_0B_n})}^2\ln{2}} & 
    -\frac{p_n^2g_n^2}{B_n^3N_0^2{(1+\frac{p_ng_n}{N_0B_n})}^2\ln{2}} 
    \end{bmatrix}
    \notag
\end{align}
Then for any vector $x=[x_1,x_2]^T \in {\mathbb{R}}^2$, we obtain $x^{T}Hx=-\frac{{(x_1B_n-x_2p_n)}^2{g_n^2}}{B_n^3N_0^2(1+\frac{p_ng_n}{N_0B_n})^2\ln{2}} \leq 0$, so $H$ is a negative semidefinite matrix, which means that $r_n(p_n,B_n)$ is a jointly concave function. Besides, $p_n \rho_nd_{n}$ is an affine function with $p_n$. Referring to Page~245 of the book \cite{cambini2008generalized}, the ratio is deemed pseudoconvex when the numerator satisfies the properties of being non-negative, concave, and differentiable, and the denominator satisfies the properties of being positive, convex, and differentiable. Thus,  Lemma \ref{lemma1} is proved. 
\end{proof}

\begin{remark}
    Based on the proof of Lemma \ref{lemma1}, $\frac{p_n \rho_nd_{n}}{r_n}$ is jointly pseudoconvex and can also be referred to as a concave-convex ratio, marked by a concave numerator and a convex denominator. Unfortunately, the sum of pseudoconvex functions does not inherently retain pseudoconvexity, presenting a challenge in solving such sum-of-ratio problems \cite{jong2012efficient,shen2018fractional1,shen2018fractional2}.
\end{remark}

To address the sum of pseudoconvex ratios problem in $\mathbb{P}_4$, a Newton-like method will be introduced in Section~\ref{section:SUB2}. Generally, our strategy entails the initial random selection of feasible $(\bm{p},\bm{B})$, followed by solving $\mathbb{P}_3$ to obtain the optimal solutions $\bm{f}^*$, $\bm{\rho}^*$ and $\mathcal{T}^*$. Subsequently, we solve $\mathbb{P}_4$ to derive the optimal $\bm{p}^*$ and $\bm{B}^*$. Optimums to $\mathbb{P}_2$ progressively converge through the iterative resolution of $\mathbb{P}_3$ and $\mathbb{P}_4$.

\section{Solutions to the joint optimization problem}
\label{Sec:Solution}
In this section, we will solve subproblems $\mathbb{P}_3$ and $\mathbb{P}_4$, and then present a resource allocation algorithm. 
\subsection{Solution to $\mathbb{P}_3$}\label{section:SUB1}

To solve problem $\mathbb{P}_3$ more efficiently, we first define $\overline{\rho}_n(p_n,B_n)$ such that
$P_n(\overline{\rho}_n,S_n(p_n,B_n))=P_n^{\textnormal{min}}$. Then constraint (\ref{constra:A}) is equivalent to
\begin{align}
\rho_n \geq  \overline{\rho}_n(p_n,B_n).  \label{constra:Aeqiv}
\end{align}
Considering constraint (\ref{constra:rho}) to the compression rate, we have:
\begin{itemize}
\item If $\overline{\rho}_n(p_n,B_n) > \rho_n^{\textnormal{max}}$,
then $\mathbb{P}_3(\boldsymbol{p},\boldsymbol{B})$ has no solution.
\item If $\overline{\rho}_n(p_n,B_n) \leq \rho_n^{\textnormal{max}}$,
then the solution $\boldsymbol{\rho}$ is given by $\boldsymbol{\rho}^{*}(\boldsymbol{p},\boldsymbol{B})$, whose $n$-th dimension is 
$\rho_n^{*}(p_n,B_n):= \max\{\overline{\rho}_n(p_n,B_n) ,\,\rho_n^{\textnormal{min}} $ according to~(\ref{constra:rho}) and~(\ref{constra:Aeqiv}).
\end{itemize}
After obtaining the optimal $\bm{\rho}^*$, we can simplify original problem $\mathbb{P}_3$ into the following:
\begin{subequations}\label{P5}
\begin{align}
\mathbb{P}_5(\boldsymbol{p},\boldsymbol{B},\boldsymbol{\rho}^{*}):
    \min_{\boldsymbol{f},\boldsymbol{h},\mathcal{T}} 
    &\left\{\omega_1 \mathcal{T} +\omega_2 \sum_{n=1}^N (E_n^{\textnormal{cmp}}+E_n^{S})\right\}
    ,\tag{\ref{P5}}  \\
    \text{s.t.}:~ 
    & \text{(\ref{constra:f}), (\ref{constra:h_sum})}, \nonumber\\  
    & \hspace{-40pt} T_n^{\text{up}}+ T_n^{\text{cmp}}(f_n)+ T_n^{S}(h_n)  \leq \mathcal{T}, ~\forall n \in \mathcal{N}. \label{constra:T}  
\end{align}
\end{subequations}

When both the objective function and constraints of a minimization problem are convex, it constitutes a convex optimization problem, leading to the following lemma:
\begin{lemma}\label{lemma:KKT}
    For a convex optimization problem, the Karush--Kuhn--Tucker (KKT) conditions are
    \begin{itemize}
        \item sufficient to find the optimum, and
        \item are necessary for optimality if at least one of the regularity conditions\footnote{Regularity conditions refer to a variety of constrained conditions under which a minimizer point $x^*$ of the original optimization problem has to satisfy the KKT conditions.} holds.
    \end{itemize}
\end{lemma}
\begin{proof}
    Please refer to chapters 5.2.3 and 5.5.3 of \cite{boyd2004convex}.
\end{proof}
It is easy to verify $\mathbb{P}_5$ is a convex optimization problem with variables $(\bm{f},\bm{h},\mathcal{T}\hspace{.8pt})$ and satisfy Slater's condition\footnote{In brief, Slater's condition means that the feasible set of the optimization problem contains at least one interior point which makes all the inequality constraints \textit{strictly} hold~\cite{boyd2004convex}.}. One viable approach to solve it is to employ the KKT conditions directly.
% Since we only optimize $\rho_n$ and the accuracy function $A_n(\rho_n,S_n)$ is concave with it, constraint (\ref{constra:A}) could be simplified into $A_n^{-1}(P_n^{\textnormal{min}},S_n) \leq \rho_n$, where $A_n^{-1}(\cdot)$ denotes the inverse function of $A_n(\cdot)$. 
Then for $\mathbb{P}_3$, with $\bm{\alpha} := [\alpha_n|_{n \in \mathcal{N}}]$ 
and  $\bm{\beta} := [\beta_n|_{n \in \mathcal{N}}]$ 
denoting the multiplier, the Lagrangian function is given by:
\begin{align}
\mathcal{L}_1&(\bm{f},\bm{h},\mathcal{T},\bm{\alpha},\bm{\tau},\bm{\beta}) = 
    \omega_1\mathcal{T} +\omega_2\sum_{n=1}^N \kappa D_n(c_{1n}f_n^2+c_{2n}h_n^2) \nonumber \\
    &
    + \sum_{n=1}^{N}  \alpha_n (f_n-f_n^{\textnormal{max}})
    +
    \sum_{n=1}^{N} \tau_n(h_n-h_n^{\text{max}})
    \nonumber\\
    &+ \sum_{n=1}^{N} \beta_n (T_n^{\text{up}}+T_n^{\text{cmp}}(f_n)+ T_n^{S}(h_n)-\mathcal{T}\hspace{.8pt}).
    % [(\frac{R_l c_n D_n}{f_n}+\frac{R_l \rho_n d_{1n}+d_{2n}}{r_n})-\mathcal{T}]
\end{align}
The KKT conditions are as follows, where we abbreviate $\mathcal{L}_1(\bm{f},\bm{h},\mathcal{T},\bm{\alpha},\bm{\nu},\bm{\beta})$ as $\mathcal{L}_1$:

\textbf{Stationarity: }
\begin{subequations}
\begin{align}
&\frac{\partial \mathcal{L}_1}{\partial f_n} = 2 \omega_2  \kappa c_{1n} D_n f_n + \alpha_n  - \beta_n \frac{c_{1n} D_n}{f_n^2} = 0, \label{Lagrange:partial_f} \\
&\frac{\partial \mathcal{L}_1}{\partial h_n} = 2 \omega_2  \kappa c_{2n} D_n h_n + \tau_n  - \beta_n \frac{c_{2n} D_n}{h_n^2} = 0, \label{Lagrange:partial_h} \\
&\frac{\partial \mathcal{L}_1}{\partial \mathcal{T}} = \omega_1 - \sum_{n=1}^{N}\beta_n = 0. \label{Lagrange:partial_T} 
% &\alpha_n \cdot (T_n^{\text{up}} + T_n^{\text{cmp}} -\mathcal{T}\hspace{.8pt}) =0. \label{Lagrange2:alpha}
\end{align}
\end{subequations}

\textbf{Complementary Slackness: }
\begin{subequations}
\begin{align}
% & \alpha_n \cdot [P_n^{\textnormal{min}} -  A_n(\rho_n)] = 0,\\
& \alpha_n \cdot (f_n - f_n^{\text{max}}) = 0, \label{Lagrange:alpha} \\
&\tau_n\cdot(h_n - h_n^{\text{max}})=0,
\label{Lagrange:nu} \\
& \beta_n \cdot (T_n^{\text{up}} + T_n^{\text{cmp}} +T_n^{S} -\mathcal{T}\hspace{.8pt}) =0. \label{Lagrange:comp_beta}
% & \delta_n \cdot (f_n^{\min} - f_n) = 0, \label{deltancdot} \\
% & \theta_n \cdot (f_n - f_n^{\textnormal{max}}) = 0, \label{thetancdot}
\end{align}
\end{subequations}

\textbf{Primal feasibility: }(\ref{constra:f}), (\ref{constra:h_sum}), (\ref{constra:Told}).

\textbf{Dual feasibility: }
\begin{subequations}\label{Dualfeasibility}
\begin{align}
\text{(\ref{Dualfeasibility}a):~}
    \alpha_n  \geq 0,~~ 
    \text{(\ref{Dualfeasibility}b):~}
    \tau_n  \geq 0,~~ \text{(\ref{Dualfeasibility}c):~}
    \beta_n  \geq 0.\notag
\end{align}
\end{subequations}
To find the optimal solution $(\bm{f}^*,\bm{h}^*,\mathcal{T}^*)$ for $\mathbb{P}_3$, we analyze the KKT conditions step by step. 
% From conditions (\ref{Lagrange:partial_f}) and (\ref{Lagrange2:comp_beta}), it is obvious that $\beta_n >0$ and we have the following equation:
% \begin{align}
% T_n^{\text{up}}+ T_n^{\text{cmp}}(f_n)=\mathcal{T}. \label{betahat}   
% \end{align}
% Solve it, and we could compute $f_n$ as a function of $\mathcal{T}$:
% \begin{align}
%     f_n = \frac{c_nD_n}{\mathcal{T}-T_n^{\text{up}}}
% \end{align}
% Here we define $\overline{f}_n(\beta_n) = \sqrt[3]{\frac{\beta_n}{2\omega_2 \kappa}}$ by setting $\alpha_n=0$ and 
% we can derive the relationship between $\mathcal{T}$ and $\beta_n$ based on (\ref{Lagrange:partial_f}):
% \begin{align}
%     \beta_n = 2 \omega_2 \kappa (\frac{c_nD_n}{\mathcal{T}-T_n^{up}})^3.
%     \label{rela:beta_T}
% \end{align}
% Substituting (\ref{rela:beta_T}) in (\ref{Lagrange:partial_T}) and we derive that:
% \begin{align}
%      \sum_{n=1}^{N} 2 \omega_2 \kappa (\frac{c_nD_n}{\mathcal{T}-T_n^{up}})^3= \omega_1
% \end{align}
% A bisection method is sufficient to solve it and derive the optimal $\mathcal{T}^*$
We first compute $f_n$ and $h_n$ as functions of $\beta_n$ by setting $\alpha_n=0$ and $\nu_n=0$ :
\begin{align}
    \overline{f}_n(\beta_n) &= \sqrt[3]{\frac{\beta_n}{2\omega_2 \kappa}},\\
    \overline{h}_n(\beta_n) &= \sqrt[3]{\frac{\beta_n}{2\omega_2 \kappa}}.
\end{align}
We first analyze $f_n$ and discuss the following two cases:
\begin{itemize}
    \item \textbf{Case 1:} $\overline{f}_n(\beta_n) \leq f_n^{\textnormal{max}}$. In this case, we simply set $\alpha_n=0$ and $f_n^*=\overline{f}_n(\beta_n)$. 
    \item \textbf{Case 2:}
    $\overline{f}_n(\beta_n) > f_n^{\textnormal{max}}$. Setting $\alpha_n = 0$ is not feasible, as it leads to $f_n^*=\overline{f}_n(\beta_n)$ from (\ref{Lagrange:partial_f}) and violates the primal feasibility constraint (\ref{constra:f}). Therefore, with $\alpha_n > 0$ in Equation (\ref{Lagrange:alpha}), we conclude $f_n = f_n^{\textnormal{max}}$.
    % $\overline{f}_n(\beta_n) > f_n^{\textnormal{max}}$. It is obvious that we cannot set $\alpha_n=0$ in this case. Otherwise, based on (\ref{Lagrange:partial_f}) the solution will be $\overline{f}_n(\beta_n)$ and the primal feasibility (\ref{constra:f}) $p_n^* \leq p_n^{\textnormal{max}}$ will be violated. Using $\alpha_n >0$ in (\ref{Lagrange:alpha}), it holds that $f_n=f_n^{\textnormal{max}}$.
\end{itemize}
Summarize the above cases, and we have:
\begin{align}
    f_n^*(\beta_n) = \min \{\overline{f}_n(\beta_n),f_n^{\textnormal{max}}\}.\label{optimal:f}
\end{align}
The analysis of $h_n$ is similar to that of $f_n$. Hence, we could also have:
\begin{align}
    h_n^*(\beta_n) = \min \{\overline{h}_n(\beta_n),h_n^{\textnormal{max}}\}.\label{optimal:h}
\end{align}
According to (\ref{Lagrange:partial_f}), only $\beta_n > 0 $ could make the equality hold. Substituting $\beta_n > 0 $ in (\ref{Lagrange:comp_beta}), we could derive that:
\begin{align}
T_n^{\text{up}}+ T_n^{\text{cmp}}(f^*_n(\beta_n))+T_n^{S}(h_n^*(\beta_n))=\mathcal{T}. \label{betahat}   
\end{align}
Based on (\ref{betahat}), we could compute $\beta_n$ as a function of $\mathcal{T}$, denoted as $\beta_n(\mathcal{T})$.
Then, we substitute $\beta_n(\mathcal{T})$ into (\ref{Lagrange:partial_T}):
\begin{align}
    \sum_{n=1}^{N} \beta_n(\mathcal{T})= \omega_1.
    \label{optimal:Told}
\end{align}
It is easy to verify that the left part of (\ref{optimal:Told}) is monotonic with respect to $\mathcal{T}$, thus a bisection method is sufficient to solve it and derive the optimal $\mathcal{T}^*$. Subsequently, $\bm{f}^*$ and $\bm{h}^*$ could also be derived. Thus, the optimal solution to $\mathbb{P}_3$ is as follows:
\begin{align}\label{optimal:rho,f,T}
\begin{cases}
    \mathcal{T}^* = \textit{Solution to (\ref{optimal:Told})}, \\
    {\rho}_n^* = \max\{\overline{\rho}_n(p_n,B_n) ,\,\rho_n^{\textnormal{min}}\},\\
    {f}^*_n= \min \{\overline{f}_n(\beta_n^*),f_n^{\textnormal{max}}\},\\
    {h}^*_n= \min \{\overline{h}_n(\beta_n^*),h_n^{\textnormal{max}}\}.
\end{cases}
\end{align}

\subsection{Solution to SUB2}\label{section:SUB2}

Problem $\mathbb{P}_4$ represents a typical sum-of-ratios problem. Given the challenges in analyzing the pseudoconvexity of the objective function $\sum_{n=1}^{N}E_n^{\text{up}}$ in $\mathbb{P}_4$, we will utilize a more efficient fractional programming method to solve it.
Initially, we use 
$r_n(p_n,B_n)$ to denote $B_n\log_2(1+\frac{p_ng_n}{N_0B_n})$ for simplicity and constraint (\ref{constra:Told}) could be rewritten as $r_n^{\textnormal{min}} \leq r_n(p_n,B_n)$, where $r_n^{\textnormal{min}} = \frac{\mathcal{T}^*-T_n^{\text{cmp}}-T_n^{S}}{\rho_nd_n}$, and constraint (\ref{constra:A}) could be simplified as ${S}_n^{\textnormal{min}} \leq \frac{p_ng_n}{N_0B_n}$, where ${S}_n^{\textnormal{min}}= A_n^{-1}(P_n^{\textnormal{min}},\rho_n^*)$. 

Furthermore, we introduce an auxiliary variable $\delta_n$ to transform problem $\mathbb{P}_4$ into an epigraph form. Through introducing an additional constraint $ \frac{p_n\rho_nd_nD_n}{r_n(p_n,B_n)} \leq \delta_n$, we can convert  $\mathbb{P}_4$ into the equivalent $\mathbb{P}_6$:
\begin{subequations}\label{Subproblem2_V1}
\begin{align}    \mathbb{P}_6(\boldsymbol{f},\bm{h},\boldsymbol{\rho}, \mathcal{T}\hspace{.8pt}):~
    &\min_{\boldsymbol{p},\boldsymbol{B},\boldsymbol{\delta}}~
     \left\{\omega_2 \sum_{n=1}^{N} \delta_n \right\},
    \tag{\ref{Subproblem2_V1}} \\
    \text{s.t.}:~ 
    & \text{(\ref{constra:p}),   (\ref{constra:B_sum})},  \nonumber \\
    & r_n^{\textnormal{min}}\! -\! r_n(p_n,B_n) \!\leq\! 0, \label{constra:r} \\
    & N_0B_nS_n^{\textnormal{min}} - p_ng_n \leq 0,\label{constra:SNR} \\
    & p_n \rho_n d_nD_n - \delta_n{r_n(p_n,B_n)} \leq 0. \label{constra:beta}
\end{align}
\end{subequations} 
However, $\mathbb{P}_6$ is still a  non-convex problem since $\delta_n{r_n(p_n,B_n)}$ in (\ref{constra:beta}) is not jointly convex. 
% The Hessian matrix for $\delta_n{r_n(p_n,B_n)}$ is not positive semidefinite (also not negative semidefinite).
Thus, we manage to further transform problem $\mathbb{P}_6$ into a convex 
problem $\mathbb{P}_7$ through the following theorem.

\begin{theorem}\label{theorem:fp method}
If $(\bm{p}^*,\bm{B}^*,\bm{\xi}^*)$ is the globally optimal solution of $\mathbb{P}_6$, then there exist $(\bm{\gamma}^*,\bm{\delta}^*)$ such that $(\bm{p}^*,\bm{B}^*)$ is a solution of the following problem for $\bm{\gamma} = \bm{\gamma}^*$ and $\bm{\delta} = \bm{\delta}^*$.
\begin{subequations}
\label{SP2_V2}
\begin{align}
\mathbb{P}_7(\boldsymbol{f},\bm{h},\boldsymbol{\rho}, \mathcal{T}):\min_{\boldsymbol{p},\boldsymbol{B}}&\left\{\!
\sum_{n=1}^{N} \!\gamma_n \!\cdot\! \big(p_n\rho_n d_nD_n \!-\! \delta_n r_n(p_n,B_n)\! \big) \!\right\}, \tag{\ref{SP2_V2}} \\
\textnormal{s.t.}:~ 
& \textnormal{(\ref{constra:p}),   (\ref{constra:B_sum}), (\ref{constra:r}), (\ref{constra:SNR})}. \nonumber
\end{align}
\end{subequations}
Furthermore, $(\bm{p}^*,\bm{B}^*)$ also satisfies the following equations under the conditions of $\bm{\gamma} = \bm{\gamma}^*$ and $\bm{\delta} = \bm{\delta}^*$:
\begin{subequations}
\begin{align}
    &\gamma_n^* = \frac{\omega_2}{r_n(p_n^*,B_n^*)},\\
    and~&\delta_n^* = \frac{p_n^*\rho_nd_nD_n}{r_n(p_n^*,B_n^*)}.
\end{align}
\end{subequations}
\end{theorem}
\begin{proof}
    The proof is given by Lemma $2.1$ in \cite{jong2012efficient}.
\end{proof}

\textbf{Theorem \ref{theorem:fp method}} indicates that problem $\mathbb{P}_7$ has the same optimal solution with $\mathbb{P}_6$.
Intuitively, we next focus on how to derive the optimal $(\bm{\delta}^*,\bm{\gamma}^*)$ to solve $\mathbb{P}_7$. We employ an alternating optimization approach to optimize the two sets of variables $(\bm{\delta},\bm{\gamma})$ and $(\bm{p},\bm{B})$, which is described in \textbf{Algorithm \ref{algo:gammadelta}}.

Algorithm \ref{algo:gammadelta} starts by calculating the initial values of auxiliary variables $(\bm{\delta}^{(0)},\bm{\gamma}^{(0)})$ based on (\ref{initial:gamma,delta}). Then, during the $i$-th iteration of Algorithm \ref{algo:gammadelta}, $(\bm{\delta}^{(i)},\bm{\gamma}^{(i)})$ are updated to $(\bm{\delta}^{(i+1)},\bm{\gamma}^{(i+1)})$ based on (\ref{Algo2:integer})-(\ref{algo:update}), which states the modified Newton method to solve $\bm{\phi}(\bm{\delta},\bm{\gamma})=0$ until convergence or reaching maximum iteration number $I$.

\begin{algorithm}
\footnotesize
\caption{Refinement of Parameters $(\bm{p},\bm{B})$}
\KwIn{Initial feasible solution $(\bm{p}^{(0)}, \bm{B}^{(0)})$;
\newline
Iteration counter $i=0$, maximum iteration number $I$; 
\newline
Error parameters $\xi \in (0,1)$, $\epsilon \in (0,1)$.}

\KwOut{Optimized $(\bm{p}^{*}, \bm{B}^{*})$.}
\label{algo:gammadelta}

\SetKwProg{MyFunction}{Function}{:}{end}
\MyFunction{Optimization\_of\_Power\_and\_Bandwidth $(\bm{p}, \bm{B}, I, \xi, \epsilon)$ }{
Compute initial values of auxiliary variables $\bm{\gamma},\bm{\delta}$:
\begin{align}
    \bm{\gamma}^{(0)} &=[\gamma_n^{(0)}|_{n \in \mathcal{N}}],~
    \bm{\delta}^{(0)} =[\delta_n^{(0)}|_{n \in \mathcal{N}}];\\
     \hspace{-10pt}where~ \delta_n^{(0)} &= \frac{p_n^{(0)}\rho_nd_nD_n}{r_n(p_n^{(0)},B_n^{(0)})},~
     \gamma_n^{(0)} = \frac{\omega_2}{r_n(p_n^{(0)},B_n^{(0)})}. \label{initial:gamma,delta}
\end{align}

\While{$i \leq I$ and not convergence}{
Solve $\mathbb{P}_7$ with given $(\bm{\delta}^{(i)},\bm{\gamma}^{(i)})$ based on (\ref{Optimal:p,B}), and obtain a new solution $(\bm{p}^{(i+1)},\bm{B}^{(i+1)})$.

% \tcp{\textit{Solve problems SUB2\_V2 based on the KKT analysis on page \pageref{Optimal:p,B} }}

Compute $\bm{\phi}(\bm{\delta}^{(i)},\bm{\gamma}^{(i)})$ as follows:
\begin{subequations}
\begin{align}
    &\bm{\phi}(\bm{\delta}^{(i)},\bm{\gamma}^{(i)})={[\bm{\phi_1}(\bm{\delta}^{(i)}),\bm{\phi_2}(\bm{\gamma}^{(i)})]}^{T} \in \mathbb{R}^{2N}, \\
    &\bm{\phi_1}(\bm{\delta}^{(i)})=[{\phi_{1,n}(\bm{\delta}^{(i)})}|_{n \in \mathcal{N}}], \\
    &\bm{\phi_2}(\bm{\gamma}^{(i)})=[{\phi_{2,n}(\bm{\gamma}^{(i)})}|_{n \in \mathcal{N}}],\\
    &\phi_{1,n}(\bm{\delta}^{(i)})\!=\!-\!p_n^{(i+1)}\rho_n d_nD_n\!+\!\delta_n^{(i)}r_n(p_n^{(i+1)},B_n^{(i+1)}),\\
    &\phi_{2,n}(\bm{\gamma}^{(i)})=-\omega_2 + \gamma_n^{(i)}r_n(p_n^{(i+1)},B_n^{(i+1)}).
\end{align}
\end{subequations}

If $\bm{\phi}(\bm{\delta}^{(i)},\bm{\gamma}^{(i)})$ is approximately a zero vector, then the algorithm terminates, and $(\bm{p}^{(i+1)},\bm{B}^{(i+1)})$ is the global optimum $(\bm{p}^{*}, \bm{B}^{*})$ to $\mathbb{P}_4$.

\tcp{\textit{The algorithm will converge when $\bm{\phi}(\bm{\delta}^{(i)},\bm{\gamma}^{(i)})=0$ and the proof is given by Theorem 3.1 in \cite{jong2012efficient}}}

Otherwise, find the smallest integer $j$ satisfying:
\begin{align}\label{Algo2:integer}
    &\parallel \bm{\phi}(\bm{\delta}^{(i)}+\xi^{j} \bm{\sigma_1}^{(i)},\bm{\gamma}^{(i)}+\xi^{j} \bm{\sigma_2}^{(i)})\parallel_2 \notag \\ &\leq (1-\epsilon\xi^{{j}})\parallel \bm{\phi}(\bm{\delta}^{(i)},\bm{\gamma}^{(i)}) \parallel_2,
\end{align}
where $\bm{\sigma_1}^{(i)}$ and $\bm{\sigma_2}^{(i)}$ are computed as
\begin{align}
    \bm{\sigma_1}^{(i)} &= -[\phi_1^{\prime}(\bm{\delta}^{(i)})]^{-1}\phi_1(\bm{\delta}^{(i)}), \\
    \bm{\sigma_2}^{(i)} &= -[\phi_2^{\prime}(\bm{\gamma}^{(i)})]^{-1}\phi_2(\bm{\gamma}^{(i)}).
\end{align}

\tcp{\textit{Here we utilize a modified Newton method, and if $j$ happens to be 0, it becomes a standard Newton Method}}

Update auxiliary variables $(\bm{\gamma},\bm{\delta})$:
\begin{align}\label{algo:update}
\hspace{-5pt}
(\bm{\delta}^{(i+1)},\bm{\gamma}^{(i+1)}) \! \leftarrow \! (\bm{\delta}^{(i)}\!+\!\xi^{j}\bm{\sigma_1}^{(i)},\bm{\gamma}^{(i)}\!+\!\xi^{j}\bm{\sigma_2}^{(i)}).
\end{align}

Let $i \leftarrow i+1$.
}
Utilize the optimized $(\bm{\delta}^{*},\bm{\gamma}^{*})$ to derive the optimal $(\bm{p}^*, \bm{B}^*)$.
}
\end{algorithm}

Next, our attention turns to solving $\mathbb{P}_7$ when  $\bm{\gamma}$ and $\bm{\delta}$ are obtained from Algorithm \ref{algo:gammadelta}. Noting that $\mathbb{P}_7$ is inherently a convex problem, KKT conditions can be utilized to solve it.
Specifically, the Lagrange function of (\ref{SP2_V2}) could be given by:
\begin{align}
    & \mathcal{L}_2 (\bm{p},\bm{B},\zeta,\bm{\eta},\bm{\nu},\bm{\iota}) = \sum_{n=1}^{N} \gamma_n\cdot\big(p_n \rho_n d_nD_n - \delta_n r_n(p_n,B_n) \big) \nonumber \\
    &+ \zeta\cdot(\sum_{n=1}^{N}B_n - B_{\textnormal{total}}) 
   + \sum_{n=1}^{N}  \eta_n \cdot (r_n^{\textnormal{min}} - r_n(p_n,B_n)) \nonumber \\
   & + \sum_{n=1}^{N}\nu_n \cdot (N_0B_nS_n^{\textnormal{min}}-p_ng_n) \!+\! \sum_{n=1}^{N} \iota_n \cdot (p_n-p_n^{\textnormal{max}}).
   \label{LagrangeSP2}
\end{align}
where $\zeta$, $\bm{\eta}$, $\bm{\iota}$ and $\bm{\nu}$ are non-negative Lagrangian multipliers.
The corresponding KKT conditions as follows:

\textbf{Stationarity: }
\begin{subequations}
\begin{align}
&\frac{\partial \mathcal{L}_2}{\partial p_n}\! = \!\gamma_n\rho_n d_nD_n \!-\!  \frac{(\gamma_n \delta_n+\eta_n) g_n}{N_0(1+S_n)\ln{2}} \!-\!\nu_ng_n \!+\! \iota_n= 0, \label{Lagrange_SP2:partial_p} \\
&\frac{\partial \mathcal{L}_2}{\partial B_n} = -(\gamma_n \delta_n+\eta_n)\log_2(1+S_n)+\frac{(\gamma_n \delta_n+\eta_n)S_n}{(1+S_n)\ln{2}} \nonumber \\
&~~~~~~~~~+ \zeta + \nu_nN_0S_n^{\textnormal{min}}= 0. \label{Lagrange_SP2:partial_B} 
\end{align}
\end{subequations}
where $S_n=\frac{p_ng_n}{N_0B_n}$ for simplicity.

\textbf{Complementary Slackness:}
\begin{subequations}
\begin{align}
&\zeta \cdot(\sum_{n=1}^{N}B_n - B) = 0, 
\label{comP:B}\\ &\eta_n\cdot(r_n^{\textnormal{min}} - r_n(p_n,B_n)) = 0, \label{comp:rn} \\
&\nu_n (N_0B_nS_n^{\textnormal{min}}-p_ng_n)
=0,\label{comp:SNR}\\
&\iota_n (p_n-p_n^{\textnormal{max}})=0.
\label{comp:p}
\end{align}
\end{subequations}

\textbf{Primal feasibility: } \text{(\ref{constra:p}),   (\ref{constra:B_sum}),  (\ref{constra:r}), (\ref{constra:SNR})}.

\textbf{Dual feasibility: }
\begin{subequations}\label{Dualfeasibility_SP2}
\begin{align}
&\text{(\ref{Dualfeasibility_SP2}a)}: \zeta \geq 0,~~
\text{(\ref{Dualfeasibility_SP2}b)}: \eta_n \geq 0, 
\nonumber\\
&\text{(\ref{Dualfeasibility_SP2}c)}: \nu_n \geq 0,~\text{(\ref{Dualfeasibility_SP2}d)}: \iota_n \geq 0. \nonumber
\end{align}
\end{subequations}

\textbf{Roadmap: } Thereafter, we identify a roadmap to find a set of solution $(\bm{p}^*,\bm{B}^*)$ by analyzing the Lagrangian multipliers with respect to the corresponding constraints.  

\textbf{Step 1.}
We first compute $\overline{p}_n$ as a function of $(B_n,\eta_n,\nu_n)$ by setting $\iota_n=0$ from (\ref{Lagrange_SP2:partial_p}), which is denoted as:
\begin{align}
 \overline{p}_n(B_n,\eta_n,\nu_n ) \!= \!
 \frac{(\gamma_n \delta_n+\eta_n) B_n}{(\gamma_n\rho_nd_nD_n\!-\!\nu_ng_n)\ln2}\!-\frac{N_0B_n}{g_n}.
\end{align}
% \begin{align}
%  \overline{p}_n(B_n,\eta_n,\nu_n ) \!= \!
%  \big(\frac{(\gamma_n \delta_n+\eta_n) g_n}{(\gamma_n\rho_nd_n\!-\!\nu_ng_n)N_0\ln2}\!-1\!\big)\frac{N_0B_n}{g_n}.
% \end{align}
Then, we discuss the following two cases to identify the optimal $p_n^*(B_n,\eta_n,\nu_n)$:
\begin{itemize}
    \item \textbf{Case 1:} $\overline{p}_n(B_n,\eta_n,\nu_n) \leq p_n^{\textnormal{max}}$. In this case, we simply set $\iota_n=0$ and the optimal $p_n^*(B_n,\eta_n,\nu_n)=\overline{p}_n(B_n,\eta_n,\nu_n)$. All conditions can be satisfied.
    \item \textbf{Case 2:}
    $\overline{p}_n(B_n,\eta_n,\nu_n) > p_n^{\textnormal{max}}$. 
Clearly, setting $\iota_n=0$ is not viable in this scenario, as it would lead to a violation of the primal feasibility condition (\ref{constra:p}). Using $\iota_n >0$ in (\ref{comp:p}), it always holds that $p_n^*=p_n^{\textnormal{max}}$. 
\end{itemize}
Summarize both cases and we have the following conclusion:
\begin{align}
    p_n^*(B_n,\eta_n,\nu_n) = \min\{\overline{p}_n(B_n,\eta_n,\nu_n),p_n^{\textnormal{max}}\}.
    \label{optimal,p}
\end{align}

\textbf{Step 2:}
Substituting (\ref{optimal,p}) into (\ref{comp:SNR}) and setting $\nu_n=0$, we could compute SNR $S_n$ as a function of $(B_n,\eta_n)$ :
\begin{align}
    \overline{S}_n(B_n,\eta_n) \!=\! \min\{\frac{(\gamma_n \delta_n+\eta_n) g_n}{(\gamma_n\rho_nd_nD_n)N_0\ln2}-1,\frac{p_n^{\textnormal{max}}g_n}{N_0B_n}\}.
\end{align}

Then, we discuss the following two cases:
\begin{itemize}
    \item \textbf{Case 1:} $\overline{S}_n(B_n,\eta_n) \geq S_n^{\textnormal{min}}$. Similarly, we could set $\nu_n=0$ and the optimal $S_n^*=\overline{S}_n(B_n,\eta_n)$. 
    \item \textbf{Case 2:}
    $\overline{S}_n(B_n,\eta_n) < S_n^{\textnormal{min}}$. Setting  $\nu_n=0$ in this case will violate primal feasibility (\ref{constra:SNR}). Thereafter, we substitute  $\nu_n >0$ in (\ref{comp:SNR}), it always holds that $S_n^*=S_n^{\textnormal{min}}$ and we can derive that $\overline{\nu}_n(B_n,\eta_n) = \frac{\gamma_n\rho_nd_n}{g_n}-\frac{(\gamma_n\delta_n+\eta_n)g_n}{(\min\{S_n^{\textnormal{min}},\frac{p_n^{\textnormal{max}}g_n}{N_0B_n}\}+1)N_0 \ln{2}}$.
\end{itemize}
Summarize both cases, and we have the following conclusion:
\begin{align}
    S_n^*(B_n,\eta_n) = \max\{\overline{S}_n(B_n,\eta_n),S_n^{\textnormal{min}}\}.
    \label{optimal,S}
\end{align}
Combining (\ref{comp:SNR}) and (\ref{optimal,S}), we can also derive the optimal $\nu_n^*(B_n,\eta_n)$ as follows:
\begin{align}\label{optimal:nu}
\nu_n^*(B_n,\eta_n) = 
\begin{cases}
    0 &\text{if $\overline{S}_n(B_n,\eta_n) \geq S_n^{\textnormal{min}}$},\\
\overline{\nu}_n(B_n,\eta_n)&\text{otherwise}.
\end{cases}
\end{align}
% \begin{align}
% &\nu_n^*(B_n,\eta_n) \text{ equals }
% 0~\text{if}~\overline{S}_n(B_n,\eta_n) \geq S_n^{\textnormal{min}}, \text{ and } \nonumber\\
% &\overline{S}_n(B_n,\eta_n)~\text{otherwise}.
%     % \nu_n^*(B_n,\eta_n) &=
%     % 0 \cdot \mathds{1}[\overline{S}_n(B_n,\eta_n) \geq S_n^{\textnormal{min}}] \nonumber \\
%     % &+ \overline{\nu}_n(B_n,\eta_n)\cdot\mathds{1} [\overline{S}_n(B_n,\eta_n) < S_n^{\textnormal{min}}],
%     \label{optimal:nu}
% \end{align}

\textbf{Step 3.} Then, we analyze the multiplier $\bm{\eta}$. Substituting (\ref{optimal,S}) into (\ref{comp:rn}), we could derive $\overline{r}_n(B_n)$ by setting $\eta_n=0$:
\begin{align}
\overline{r}_n(B_n) = B_n\log_2(1+
    S_n^*(B_n,\eta_n)\mid_{\eta_n=0}),
\end{align}
We discuss the following two cases:
\begin{itemize}
    \item \textbf{Case 1:} $\overline{r}_n(B_n) \geq r_n^{\textnormal{min}}$. We set $\eta_n=0$ and the optimal $r_n^*(B_n)=\overline{r}_n(B_n)$. 
    \item \textbf{Case 2:}
    $\overline{r}_n(B_n) < r_n^{\textnormal{min}}$. Setting  $\eta_n=0$ will violate primal feasibility (\ref{constra:r}). Using $\eta_n >0$ in (\ref{comp:rn}), it always holds that $r_n^*(B_n,\eta_n)=r_n^{\textnormal{min}}$. Solve the above equation and obtain the optimal solution denoted as $\overline{\eta}_n(B_n)$.
\end{itemize}
Summarize both cases and we have the following conclusion:
\begin{align}\label{optimal:eta}
\eta_n^*(B_n) = 
\begin{cases}
    0 &\text{if $\overline{r}_n(B_n) \!\geq r_n^{\textnormal{min}}$},\\
\overline{\eta}_n(B_n)&\text{otherwise}.
\end{cases}
\end{align}
% \begin{align}\label{optimal:nu}
% \nu_n^*(B_n,\eta_n) = 
% \begin{cases}
%     0 &\text{if $\overline{S}_n(B_n,\eta_n) \geq S_n^{\textnormal{min}}$},\\
% \overline{\nu}_n(B_n,\eta_n)&\text{otherwise}.
% \end{cases}
% \end{align}
% \begin{align}
% \eta_n^*(B_n) \text{ equals }
%     0~\text{if}~\overline{r}_n(B_n) \!\geq r_n^{\textnormal{min}}, \text{ and } \overline{\eta}_n(B_n)~\text{otherwise}.
%     % \eta_n^*(B_n) \!=\!
%     % 0 \! \cdot \!\mathds{1}[\overline{r}_n(B_n) \!\geq r_n^{\textnormal{min}}] \!+ \!\overline{\eta}_n(B_n)
%     % \! \cdot \!\mathds{1}[\overline{r}_n(B_n)\! < \! r_n^{\textnormal{min}}],
%     \label{optimal:eta}
% \end{align}

\textbf{Step 4.} Substituting (\ref{optimal,p}), (\ref{optimal:nu}), (\ref{optimal:eta}) into stationarity (\ref{Lagrange_SP2:partial_B}), the following equation could be derived:
\begin{align}
    &-(\gamma_n \delta_n+\eta_n^*(B_n))\log_2(1+S_n^*(B_n))+
    \nonumber \\
    &\frac{(\gamma_n \delta_n+\eta_n^*(B_n))S_n^*(B_n)}{(1+S_n^*(B_n))\ln{2}} 
    + \zeta + \nu_n^*(B_n)N_0S_n^{\textnormal{min}}= 0. \label{optimal:B}
\end{align}
From (\ref{optimal:B}) we could derive the solution denoted as $\overline{B}_n(\zeta)$. We calculate $\sum_{n=1}^{N}\overline{B}_n(0)$ with $\zeta=0$, and then considering:
\begin{itemize}
    \item \textbf{Case 1:} $\sum_{n=1}^{N}\overline{B}_n(0) \leq B_{\textnormal{total}}$. Set $\zeta=0$ and the optimal $B_n^*=B_n(0)$. 
    \item \textbf{Case 2:}
    $\sum_{n=1}^{N}\overline{B}_n(0) > B_{\textnormal{total}}$. Using $\zeta >0$ in (\ref{comP:B}) and we have $\sum_{n=1}^{N}\overline{B}_n(\zeta) = B_{\textnormal{total}}$. A bisection method could be applied to solve the equation and find the solution $\overline{\zeta}$.
\end{itemize}
Summarize both cases and we have:
\begin{align}\label{optimal:zeta}
\zeta^* = 
\begin{cases}
    0 &\text{if $\sum_{n=1}^{N}\!\overline{B}_n(0) \!\leq \! B_{\textnormal{total}}$},\\
\overline{\zeta}&\text{otherwise}.
\end{cases}
\end{align}
% \begin{align}
%     \zeta^* \text{ equals }
%     0~\text{if}~\textstyle{\sum_{n=1}^{N}\!\overline{B}_n(0) \!\leq \! B_{\textnormal{total}}}, \text{ and } \overline{\zeta}~\text{otherwise}.
%     \label{optimal:zeta}
% \end{align}
Until now, we have analyzed all the variables, and the optimal solution $(\bm{p}^*,\bm{B}^*)$ could be expressed as:
\begin{align}\label{Optimal:p,B}
\begin{cases}
    {B}_n^* = B_n(\zeta^*),\\
    {p}^*_n= \min\{\overline{p}_n(B_n^*,\eta_n^*,\nu_n^*),p_n^{\textnormal{max}}\}. 
\end{cases}
\end{align}
\subsection{Resource Allocation Algorithm}

Based on the analysis in Sections \ref{section:SUB1} and \ref{section:SUB2}, we have solved Problems $\mathbb{P}_3$ and $\mathbb{P}_4$ respectively. Thereafter, we give the resource allocation algorithm, as shown in Algorithm \ref{algo:resourceallocation} below, to iteratively solve $\mathbb{P}_3$ and $\mathbb{P}_4$ until convergence.

% \begin{algorithm}
% \footnotesize
% \caption{Refinement of Parameters $(\bm{p},\bm{B})$}
% \KwIn{Initial feasible solution $(\bm{p}^{(0)}, \bm{B}^{(0)})$;
% \newline
% Iteration counter $i=0$, maximum iteration number $I$; 
% \newline
% Error parameters $\xi \in (0,1)$, $\epsilon \in (0,1)$.}

% \KwOut{Optimized $(\bm{p}^{*}, \bm{B}^{*})$.}
% \label{algo:gammadelta}

% \SetKwProg{MyFunction}{Function}{:}{end}
% \MyFunction{\underline{Optimization\_of\_Power\_and\_Bandwidth $(\bm{p}, \bm{B})$ }}{

\begin{algorithm}
\footnotesize
\caption{Resource Allocation Algorithm} 
\label{algo:resourceallocation}
\KwIn{Initial solution $sol^{(0)}=(\bm{p}^{(0)},\bm{B}^{(0)},\bm{f}^{(0)},\bm{h}^{(0)},\bm{\rho}^{(0)})$;
\newline
Maximum iteration number $K$; 
\newline
Error parameter $\epsilon \in (0,1)$.}

\KwOut{Optimal solution $sol^{*}=(\bm{p}^{*}, \bm{B}^{*},\bm{f}^{*},\bm{h}^{*},\bm{\rho}^{*})$.}

\SetKwProg{MyFunction}{Function}{:}{end}

\MyFunction{Joint\_Optimization $(sol^{(0)},K,\epsilon)$}{

Set iteration counter $k=0$.

\While{$k \leq K$ and not convergence}{

Solve $\mathbb{P}_3$ through (\ref{optimal:rho,f,T}) on page \pageref{optimal:rho,f,T} with $(\bm{p}^{(k)},\bm{B}^{(k)})$, and obtain $(\bm{f}^{(k+1)},\bm{h}^{(k+1)},\bm{\rho}^{(k+1)},\mathcal{T}^{(k+1)})$.

Solve $\mathbb{P}_4$ by calling Algorithm \ref{algo:gammadelta} on page \pageref{algo:gammadelta} with $(\bm{f}^{(k+1)},\bm{h}^{(k+1)},\bm{\rho}^{(k+1)},\mathcal{T}^{(k+1)})$ to obtain $(\bm{p}^{(k+1)},\bm{B}^{(k+1)})$. 

% \tcp{\textit{Solve problems $\mathbb{P}_3$ and $\mathbb{P}_4$ alternatively until convergence}}

Update the solution:
\begin{align}
sol^{(k+1)} \leftarrow (\bm{p}^{(k+1)},\bm{B}^{(k+1)},\bm{f}^{(k+1)},\bm{h}^{(k+1)},\bm{\rho}^{(k+1)}). \nonumber
\end{align}

If $|sol^{(k+1)}-sol^{(k)}| \leq \epsilon $, the algorithm terminates.

Let $k \leftarrow k+1$
}
Return the optimal $sol^{*}=(\bm{p}^{*}, \bm{B}^{*},\bm{f}^{*},\bm{h}^{*},\bm{\rho}^{*})$
}
\end{algorithm}

\section{Time complexity, solution quality and convergence analysis}\label{Sec:time Complexity}

This section offers a detailed analysis of Algorithm \ref{algo:resourceallocation}, focusing on its time complexity, the quality of solutions, and its convergence 

\subsection{Time Complexity}
The bulk of the computational load in Algorithm \ref{algo:resourceallocation} resides in steps 4 and 5. Step 4 involves using the bisection method to solve $\mathbb{P}_3$ and to compute the parameters $\bm{f},\bm{h},\bm{\rho}$ and $\mathcal{T}$. Given that the bisection method's complexity is independent of $N$ and hinges only on the desired precision, step 3 incurs a complexity of $\mathcal{O}(4N)$. In step 5, Algorithm \ref{algo:gammadelta} iteratively solves $\mathbb{P}_7$ and updates auxiliary variables until the convergence, where the maximum iteration number is $I$.
The computation of $\bm{B}$ and $\bm{p}$, along with the updates to $\bm{\delta}$ and $\bm{\gamma}$, takes $\mathcal{O}(4N)$. Additional steps 5, 7, and 8 in Algorithm \ref{algo:gammadelta} incur complexities of $\mathcal{O}(2N)$, $\mathcal{O}((j+1)N)$, and $\mathcal{O}(2N)$, respectively. Considering a maximum iteration count of $I$, the overall complexity for Algorithm \ref{algo:gammadelta} is thus approximated as $\mathcal{O}\big(I(j+9)N\big)$. Accounting for the fact that Algorithm \ref{algo:resourceallocation} iteratively solves $\mathbb{P}_3$ and $\mathbb{P}_4$ with a maximum of $K$ iterations, we deduce that the total time complexity is $\mathcal{O}\big(K(Ij+9I+4)N\big)$.

\subsection{Solution Quality}

The proposed Algorithm \ref{algo:resourceallocation} iteratively optimizes two variable sets, $[\bm{f},\bm{h},\bm{\rho},\mathcal{T}]$ and $[\bm{p},\bm{B}]$, by solving problems $\mathbb{P}_3$ and $\mathbb{P}_4$ alternately. For $\mathbb{P}_3$, Lemma \ref{lemma:KKT} demonstrates that the KKT conditions are both necessary and sufficient for optimality due to the problem's convexity. This affirmation means that a globally optimal solution can be reliably obtained. Moreover, Algorithm \ref{algo:gammadelta} aims to solve $\mathbb{P}_4$, and based on Theorem \ref{theorem:fp method},
$\mathbb{P}_7$ is proven to have the same globally optimal solution as $\mathbb{P}_4$. Therefore, although the alternating optimization approach does not always assure a local or global optimum for \(\mathbb{P}_1\), the global optimality of the subproblems can be secured. The precision parameter \(\epsilon\) in Algorithm \ref{algo:resourceallocation} guarantees the solution's convergence to a desired level of accuracy.

\subsection{Convergence}
Convergence begins with the resolution of $\mathbb{P}_3$ using the KKT conditions within Algorithm \ref{algo:resourceallocation}. Given the convex nature of $\mathbb{P}_3$ when $\bm{p}$ and $\bm{B}$ are fixed, we ensure both the optimality and convergence, as dictated by Lemma \ref{lemma:KKT}.

The convergence of $\mathbb{P}_4$ solved by Algorithm \ref{algo:gammadelta} merits further discussion. Let $\bm{\upsilon} = (\bm{\delta},\bm{\gamma})$ for brevity. Drawing on Theorem~3.2 in \cite{jong2012efficient}, convergence is assured if:
\begin{enumerate}
\item The function $\bm{\phi}(\bm{\upsilon})$ is differentiable and meets the Lipschitz condition within the solution domain $\Omega$.
\item A constant $L > 0$ exists, such that $| \bm{\phi}'(\bm{\upsilon}_1) - \bm{\phi}'(\bm{\upsilon}_2) | \leq L | \bm{\upsilon}_1 - \bm{\upsilon}_2 |$ for all $\bm{\upsilon}_1, \bm{\upsilon}_2 \in \Omega$.
\item A constant $M > 0$ can be found for which $| [\bm{\phi}'(\bm{\upsilon})]^{-1} | \leq M$ holds true for all $\bm{\upsilon} \in \Omega$.
\end{enumerate}

Under these stipulations, the modified Newton method used in Algorithm \ref{algo:gammadelta} converges linearly to the optimal solution $\bm{\upsilon}^*$ from any starting point $\bm{\upsilon}_0 \in \Omega$, and this convergence becomes quadratic in the vicinity of $\bm{\upsilon}^*$. The function $\bm{\phi}(\bm{\upsilon})$ is linear and readily satisfies conditions 1)-3), as shown in \cite{jong2012efficient}. With $\bm{\upsilon}$ given, the KKT conditions are again employed to determine the optimal $(\bm{p}^{*},\bm{B}^{*})$. This analysis substantiates the convergence of Algorithm \ref{algo:gammadelta}.

\section{Simulation}\label{Sec:Simulation}
This section first introduces the adopted SemCom models, accuracy functions and parameter settings. 
Then, experiments are conducted to demonstrate our method's effectiveness.
\subsection{Semantic Communication Model}

We utilize a deep joint source and channel coding (JSCC) model, as outlined in~\cite{bourtsoulatze2019deep}, for our image semantic communication model. 
The deep JSCC architecture has been widely applied in semantic communication studies~\cite{xu2023deep,park2024joint,wu2023cddm}.
The model functions as an end-to-end autoencoder, consisting of two convolutional neural networks (CNNs) that serve as the encoder and decoder. These components are jointly trained with a non-trainable layer in the middle that simulates the noisy communication channel.

In particular, the encoder CNN consists of five convolutional layers followed by parametric ReLU (PReLU) activation functions and a normalization layer. This architecture extracts and combines image features from the original image to form the channel input samples for transmission. Conversely, the decoder CNN processes the received coded signal (possibly noised) through a series of transposed convolutional layers with non-linear activation functions, using these features to reconstruct an estimated image.

% The semantic encoder employs two convolutional layers with ReLU activations to extract semantic features, followed by max-pooling for spatial down-sampling, culminating in a compressed feature map. Then, a dense layer, acting as a channel encoder, maps the feature to a lower dimension according to a specific compression rate. After encoding, the representation undergoes AWGN noise simulation. The channel decoder resizes the representation to a uniform size and then sends it to the semantic decoder for image reconstruction. The classifier is also a CNN-based neural model.

\subsection{Performance Function Fitting}\label{section:accuracy_function}

We evaluate the PSNR performance of the adopted model on the CIFAR-10 dataset, which is aligned with the choice of deep JSCC~\cite{bourtsoulatze2019deep}. Specifically, CIFAR-10 consists of color images with a $32 \times 32$ pixels resolution.

\begin{figure*}
    \centering
    \begin{subfigure}[b]{0.45\linewidth}
        \centering
        \includegraphics[width =\linewidth]{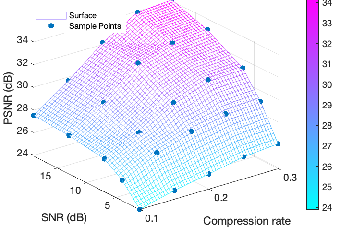}
        \caption{Interploted surface for CIFAR-10.}
    \end{subfigure}
    \hspace{15pt}
    \begin{subfigure}[b]{0.45\linewidth}
        \centering
        \includegraphics[width =\linewidth]{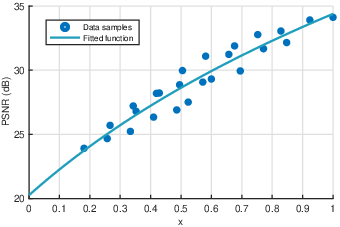}
        \caption{Fitted function for CIFAR-10.}
    \end{subfigure}
    \caption{Interploted surfaces and fitted results for CIFAR-10. Specifically, each data point (i.e., different SNR and compression rate) was run ten times, and the average accuracy was taken to reduce experimental error.}
    \label{fig:Inteploted _surface}
\end{figure*}

% \begin{figure}
%     \centering
%     \includegraphics[width= 0.7\linewidth] {figs/Fitting_functions}
%     \caption{Modelling functions CIFAR-10}
%     \label{fig:accuracy_function}
% \end{figure}
% \textbf{SemCom Model.}

\textbf{Performance function.} 
The PSNR performance function $P(\rho_n, S_n)$ for each user $n$ has two variables: compression rate $\rho_n$ and SNR $S_n$. A higher $\rho_n$ improves accuracy by minimizing information loss in compression, while a higher SNR ensures better semantic content quality and reliability. We employ curve-fitting techniques to develop function $P(\rho_n, S_n)$, and the specific expression is adapted from logarithmic models in previous works~\cite{zhao2023optimizing}, reflecting diminishing returns:
\begin{align}
    P(\rho_n,S_n) = a\ln(c^{\rho}\rho_n+c^{s}S_n+b),
    % 1 - \exp \big({-(c_n^{\rho}\rho_n +c_n^{S}S_n)}\big)
\end{align}
where $a,b,c^{\rho}$ and $ c^{s}$ are coefficients decided by fitting data.

Given the multidimensional nature of $P(\rho_n, S_n)$, we simplify visual interpretation by transforming it into a two-dimensional plot. Let $\rho_n^{\textnormal{max}}$ (resp., $S_n^{\textnormal{max}}$) be the maximum $\rho_n$ (resp., $S_n$) from the dataset. Upon introducing the notation $l = c^{\rho}\rho_n^{\textnormal{max}} + c^{s}S_n^{\textnormal{max}} $, we proceed to assign $\frac{c^{\rho}}{l}\rho_n + \frac{c^{s}}{l}S_n$ as the $x$-axis coordinate, and plot $a\ln(lx+b)$ as the $y$-axis coordinate. This choice is motivated by the fact that $P(\rho_n, S_n) = a\ln(c^{\rho}\rho_n+c^{s}S_n+b) = a\ln(lx+b)$, thus allowing for a direct representation of the function. Following the mentioned transformation, the $x$-coordinate of each data point now falls within the range of 0 and 1.

\textbf{Modelling results.} 
We use the linear method in \textit{scatteredInterpolant} function of Matlab for interpolation. Compared with other interpolation methods (e.g., ``nearest'' and ``natural''), linear interpolation provides smooth transitions between data points without excessive oscillations or discontinuities, and the overall data trend is also maintained well. In addition, we also employ a 5-fold cross-validation to ensure the interpolation accuracy. The entire procedure is repeated ten times and the average mean absolute error (MAE) errors on all folds are 0.4381.
The interpolation and function fitting performance and the fitted function are shown in Fig.~\ref{fig:Inteploted _surface}. 
The expression for the fitted function in Fig.~\ref{fig:Inteploted _surface}(b) is further stated as follows:
\begin{align}
    P(x)\!=\! 18.67\ln(3.35x+5.11), \textnormal{ for } x\!=\! 1.52\rho\!+\!0.03S.
\end{align}

% \begin{table}[!ht]
%     \centering
%     \caption{Accuracy Function Simulation}
%     \label{tab:function_fitting}
%     \begin{tabular}{l|l}
%     \hline
%         Dataset & \begin{tabular}[c]{@{}l@{}}
%         Utility Function from curve-fitting  for  normalized  $x$
%         \end{tabular}\\ \hline
%         MNIST & 
%         $1-0.9168e^{-5.4490x}$, for $x = 0.4058\rho+0.1671S$ 
%           \\ \hline
%         CIFAR-10&  
%         $1-0.7396e^{-0.9438x}$, for $x = 0.4618\rho+0.0365S$.
%           \\ \hline
%     \end{tabular}
% \end{table}

\subsection{Parameter Settings}\label{section:settings}

Our experiments employ a foundational set of parameters as defaults, primarily adapted from~\cite{yang2020energy}. Table~\ref{tab:sys_para} summarizes the overall system settings. The channel model is devised for urban landscapes, encompassing a path loss calculated by $128.1+37.6\log(\text{distance})$ with an $8$ dB standard deviation. 
Gaussian noise power spectral density
$N_0=-174$ dBm/Hz represents thermal noise at 20$^{\circ}$C.

Users are randomly dispersed within a circular area with a 500-meter diameter centered on the base station. The system contains $50$ users and a total bandwidth of $20$ MHz. The effective switched capacitance (\(\kappa\), $10^{-28}$) and the sample size in a batch ($D_n$, 32) are preset. CPU cycle requirements (\(c_{1n}\)) and (\(c_{2n}\)) randomly range from $[1,3] \times 10^6$ and $[3,5] \times 10^6$, respectively.
The data size (\(d_n\)) is $1$ Mbits. The maximum transmission power (\(p_n^{\textnormal{max}}\)) and device computing frequency (\(f_n^{\textnormal{max}}\)) are 20 dBm and 1 GHz. The maximum allocated frequency to device $n$ at BS (\(h_n^{\textnormal{max}}\)) is 5 GHz. The PSNR threshold \(P_n^{\textnormal{min}}\) is initialized at 25 dB. The semantic compression ratio $\rho$ is between ($\rho_n^{\textnormal{min}}$, 0.1) and ($\rho_n^{\textnormal{max}}$, 0.3).

\begin{table}[ht!]
\caption{System Parameter Settings}\label{tab:sys_para}
\begin{tabular}{ll}
\hline
\textbf{Parameter}                                     & \textbf{Value}   \\ \hline
The path loss model & $128.1\!+\!37.6\log(\text{d})$ \\
The standard deviation of shadow fading   & $8$ dB \\
Noise power spectral density $N_0$  & $-174$ dBm/Hz \\
The number of users $N$  & $50$  \\
Total Bandwidth $B$  & $20$ MHz  \\
Effective switched capacitance $\kappa$ & $10^{-28}$ \\
The number of samples $D_{n}$& $32$  \\
CPU cycles on each device $c_{1n}$   &  $[1,3] \times 10^6$ \\
CPU cycles on the BS $c_{2n}$  &  $[3,5] \times 10^6$ \\
The uploaded data size $d_{n}$ & $1$ Mbits     \\
Max. uplink transmission power $p_n^{\textnormal{max}}$ & $20$ dBm   \\
Max. device computing frequency $f_n^{\textnormal{max}}$ & $1$ GHz  \\
Max. BS frequency allocated to device $h_n^{\textnormal{max}}$ & $5$ GHz  \\
The PSNR threshold $P_n^{\textnormal{min}}$  & $25$ dB  \\
Min. semantic compression ratio $\rho_n^{\textnormal{min}}$ & $0.1$ \\
Max. semantic compression ratio $\rho_n^{\textnormal{max}}$ & $0.3$ \\
\hline
\end{tabular}
\end{table}

\subsection{Comparison of Different Baselines}
We first compare our proposed Algorithm \ref{algo:resourceallocation} with the following $4$ baselines:

\begin{figure*}
    \centering
    \begin{subfigure}{0.30\linewidth}
        \centering
        \includegraphics[width =\linewidth]{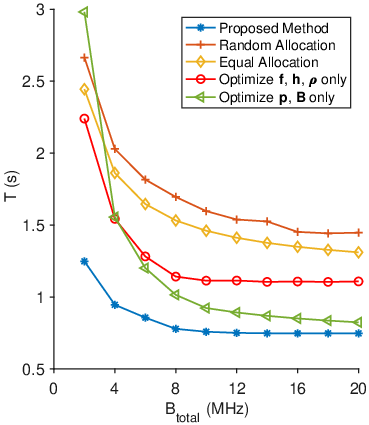}
        \caption{Total time $T$.}
    \end{subfigure}
    \hspace{10pt}
    \begin{subfigure}{0.30\linewidth}
        \centering
        \includegraphics[width =\linewidth]{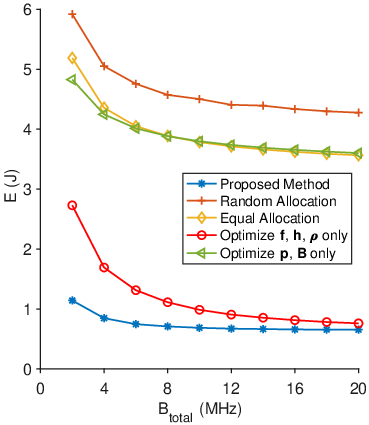}
        \caption{Total energy $E$.}
    \end{subfigure}
    \hspace{10pt}
    \begin{subfigure}{0.30\linewidth}
        \centering
        \includegraphics[width =\linewidth]{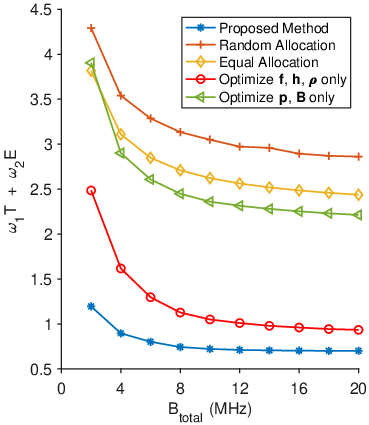}
        \caption{$\omega_1T+\omega_2E$.}
    \end{subfigure}
    \caption{Time, energy, and their weighted sum under different total bandwidth $B_{\textnormal{total}}$.}
    \label{fig:max_B}
\end{figure*}

\begin{itemize}
    \item[(i)] 
    \textbf{Random Allocation:} For each device $n$, the quantities $p_n, f_n, h_n, \rho_n$ take  random values from $[0,p_n^{\textnormal{max}}]$, $[0,f_n^{\textnormal{max}}]$, $[0,h_n^{\textnormal{max}}]$, $[\rho_n^{\min},\rho_n^{\textnormal{max}}]$ respectively. A truncated normal distribution is employed for this purpose, centered at the range's midpoint and with a standard deviation set to one-tenth of the range width.
    Besides, $B_n$ is selected from $[\frac{B_{\textnormal{total}}}{1.25N},\frac{B_{\textnormal{total}}}{0.8N}]$ and the sum of $B_n$ is equal to $B_{\textnormal{total}}$.
    \item[(ii)] \textbf{Average Allocation:} Here we set each $B_n$ as $\frac{B_{\textnormal{total}}}{N}$, and $p_n, f_n, h_n, \rho_n$ as $\frac{p_n^{\textnormal{max}}}{2}, \frac{f_n^{\textnormal{max}}}{2}, \frac{h_n^{\textnormal{max}}}{2}, \frac{\rho_n^{\textnormal{max}}+\rho_n^{\textnormal{min}}}{2}$ respectively to evenly distribute resources. 
    \item[(iii)] \textbf{Optimize $\bm{p},\bm{B}$ only:} Here we fix each $f_n,h_n$ and $\rho_n$ as the default settings in Average Allocation. Then, we optimize $p_n$ and $B_n$ of each device $n$. 
    \item[(iv)] \textbf{Optimize $\bm{f},\bm{h},\bm{\rho}$ only:} Each $p_n$ and $B_n$ are fixed as in Average Allocation, and we only optimize$f_n,h_n$ and $\rho_n$ of each device $n$. 
\end{itemize}

Then, we compare the performance between our proposed method and baselines under varying resource conditions. To reduce the interference of randomness, we choose data points uniformly and run each experiment 100 times to take the average value as the final result.

\begin{figure*}
    \centering
    \begin{subfigure}{0.30\linewidth}
        \centering
        \includegraphics[width =\linewidth]{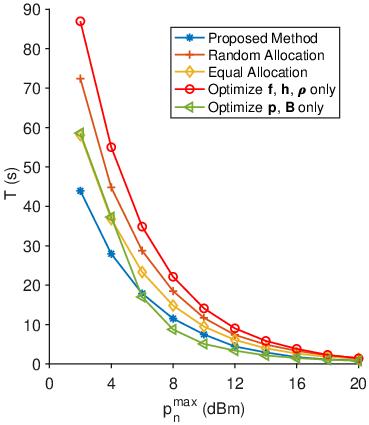}
        \caption{Total time $T$.}
    \end{subfigure}
    \hspace{10pt}
    \begin{subfigure}{0.30\linewidth}
        \centering
        \includegraphics[width =\linewidth]{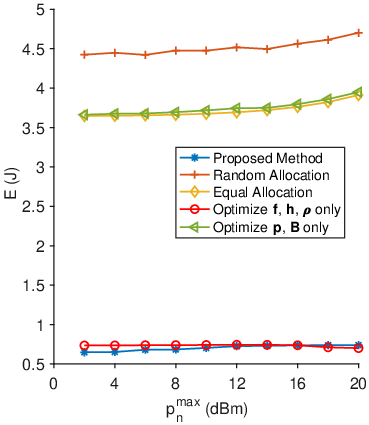}
        \caption{Total energy $E$.}
    \end{subfigure}
    \hspace{10pt}
    \begin{subfigure}{0.30\linewidth}
        \centering
        \includegraphics[width =\linewidth]{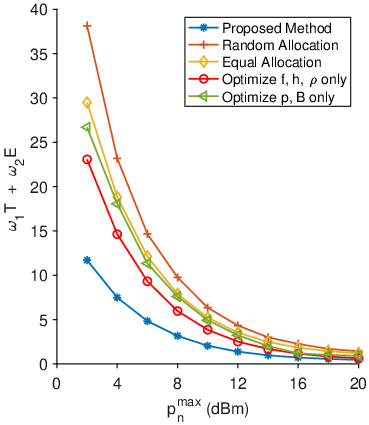}
        \caption{$\omega_1T+\omega_2E$.}
    \end{subfigure}
    \caption{Time, energy, and their weighted sum under different maximum transmission power $p_n^{\textnormal{max}}$.}
    \label{fig:max_p}
\end{figure*}

\begin{figure*}[h]
    \centering
    \begin{subfigure}{0.30\linewidth}
        \centering
        \includegraphics[width =\linewidth]{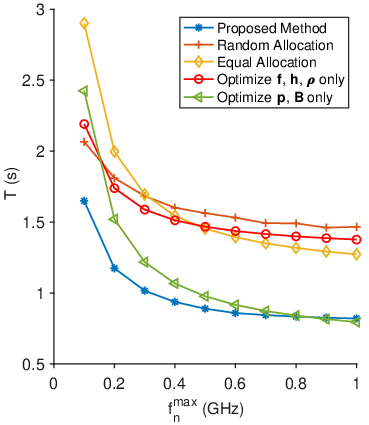}
        \caption{Total time $T$.}
    \end{subfigure}
    \hspace{10pt}
    \begin{subfigure}{0.30\linewidth}
        \centering
        \includegraphics[width =\linewidth]{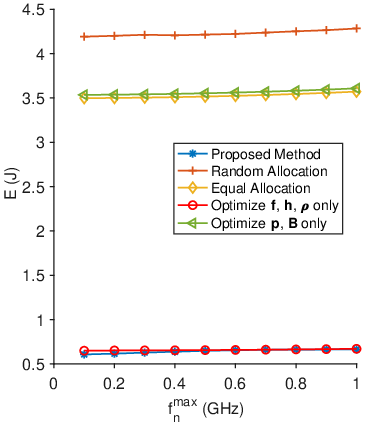}
        \caption{Total energy $E$.}
    \end{subfigure}
    \hspace{10pt}
    \begin{subfigure}{0.30\linewidth}
        \centering
        \includegraphics[width =\linewidth]{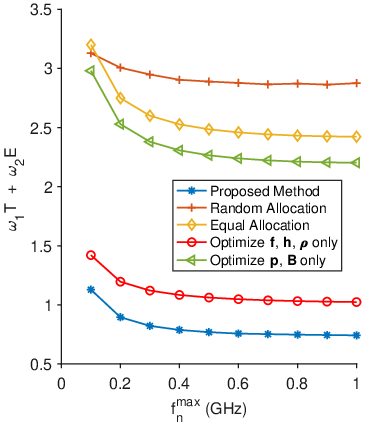}
        \caption{$\omega_1T+\omega_2E$.}
    \end{subfigure}
    \caption{Time, energy, and their weighted sum under different maximum frequency $f_n^{\textnormal{max}}$.}
    \label{fig:max_f}
\end{figure*}

\textbf{Total Bandwidth.} Here we vary the total available bandwidth $B_{\textnormal{total}}$ from $1$ MHz to $20$ MHz. Figs.~\ref{fig:max_B}(a),(b), and (c) demonstrate the total time completion $T$, energy consumption $E$, and weighted sum $\omega_1T+\omega_2E$ with respect to $B_{\textnormal{total}}$. Higher total bandwidth means more available communication resources so that the general decreasing trend can be observed in the curves of each figure. 

Specifically, in Fig.~\ref{fig:max_B}(a), we can see 
our proposed method stays optimal under different $B_{\textnormal{total}}$. Besides, ``Optimize $\bm{p},\bm{B}$ only" outperforms ``Optimize $\bm{f},\bm{h},\bm{\rho}$ only" as $B_{\textnormal{total}}$ increases.
% It is worth noting that 
% ``Optimize $\bm{f},\bm{h},\bm{\rho}$ only" is better at first and is later overtaken by ``Optimize $\bm{p},\bm{B}$ only" as $B_{\textnormal{total}}$ increases. 
We analyze that it is because ``Optimize $\bm{p},\bm{B}$ only" could bring better communication optimization in the case of abundant bandwidth resources. 
Fig.~\ref{fig:max_B}(b) presents the overall energy consumption. The computation in training incurs more energy overhead compared to transmission. Hence, our proposed method and ``Optimize $\bm{f},\bm{h},\bm{\rho}$ only" could obtain much better performance on energy optimization. The performance of our method on $E$ is also better than other benchmarks. 
In Fig.~\ref{fig:max_B}(c), our method outperforms all the baselines under different $B_{\textnormal{total}}$.
% Additionally, it is worth noting that ``Optimize $\bm{p},\bm{B}$ only" is better at first and is later overtaken by ``Optimize $\bm{f},\bm{\rho}$ only" as $B_{\textnormal{total}}$ increases. 

\textbf{Maximum Transmission Power.} We vary the maximum transmission power $p_n^{\textnormal{max}}$ from $4$ dBm to $20$ dBm, and show the experimental results in Figs.~\ref{fig:max_p}(a),(b) and (c). Higher $p_n^{\textnormal{max}}$ indicates lower limits on the energy consumption $E$; hence, the algorithm could utilize higher transmission power to optimize time consumption $T$. 

In Fig.~\ref{fig:max_p}(a), all curves show a downward trend on $T$, which is consistent with our analysis above. It is worth noting that although ``Optimize $\bm{p},\bm{B}$ only" sometimes could achieve smaller $T$ than ours, the gap is much smaller than that of $E$.  
We can observe a moderate rise of $E$ in Fig.~\ref{fig:max_p}(b). This is because as the range of $p_n^{\textnormal{max}}$ expands, more optimal solutions $\bm{p}$ could be chosen to bring about better joint optimization by sacrificing energy optimization to some extent. Finally, in Fig.~\ref{fig:max_p}(c), the proposed method still has the best joint optimization performance, proving the superiority of our method. 

\textbf{Maximum Device Computation Frequency.} We vary the maximum CPU frequencies $f_n^{\textnormal{max}}$ from $0.1$ GHz to $1$ GHz and compare $T$ and $E$ in Fig.~\ref{fig:max_f}(a) and Fig.~\ref{fig:max_f}(b). Intuitively, higher CPU frequency means shorter computation latency but higher energy consumption. Thus, there is a downward trend for curves in Fig.~\ref{fig:max_f}(a) and vice versa in Fig.~\ref{fig:max_f}(b), respectively. It is also notable that our proposed method and ``Optimize $\bm{f},\bm{\rho}$ only" will enter a stationary phase as the $f_n^{\textnormal{max}}$ continuously increase. This is because the optimal frequencies have already been found. In Fig.~\ref{fig:max_f}(c), it is evident that our proposed method always remains the best. 

The above experiments demonstrate our proposed method's superiority over benchmarks across various resource scenarios. 

\begin{figure}
    \centering
    \begin{subfigure}[b]{0.49\linewidth}
        \centering
        \includegraphics[width =\linewidth]{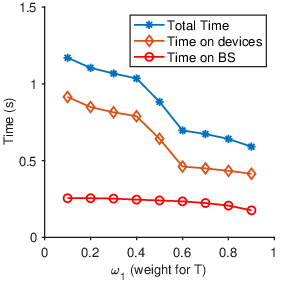}
        \caption{Time.}
    \end{subfigure}
    \hfill
    \begin{subfigure}[b]{0.49\linewidth}
        \centering
        \includegraphics[width =\linewidth]{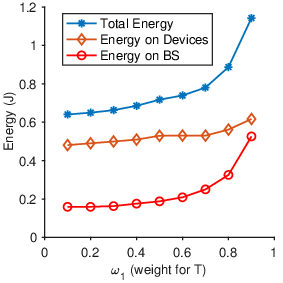}
        \caption{Energy.}
    \end{subfigure}
    \caption{Time and energy under different weight parameters $\omega_1$.}
    \label{fig:omega}
\end{figure}

\begin{figure}
    \centering
    \begin{subfigure}[b]{0.49\linewidth}
        \centering
        \includegraphics[width =\linewidth]{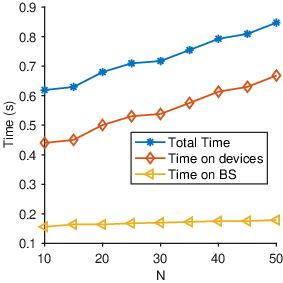}
        \caption{Time.}
    \end{subfigure}
    \hfill
    \begin{subfigure}[b]{0.49\linewidth}
        \centering
        \includegraphics[width =\linewidth]{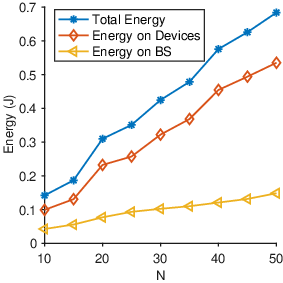}
        \caption{Energy.}
    \end{subfigure}
    \caption{Time and energy under different numbers of users $N$.}
    \label{fig:number_of_users}
\end{figure}

\subsection{Weight Parameters}

Here, we study how values of weight parameters influence our algorithm's performance. In our problem, there are two weight parameters $\omega_1$ and $\omega_2$, and the sum of them equals $1$. If the system is more sensitive to latency, the parameter $\omega_1$ for $T$ should be set as a larger value. Then, if the system has higher requirements for low energy consumption, $\omega_2$ for $E$ should be raised accordingly. 

We set the parameter $\omega_1$ within the range of $0.1$ to $0.9$, while $\omega_2$ is configured as $1-\omega_1$. The experimental result is reported in Figs.~\ref{fig:omega}(a) and (b). In Fig.~\ref{fig:omega}(a), we can see as $\omega_1$ increases, the time consumption continues to decrease. This is due to a larger $\omega_1$ could motivate our algorithm to allocate more resources to the optimization of time.
Similarly, a decrease in $\omega_2$ represents less emphasis on energy optimization, so the total energy consumption $E$ in Fig.~\ref{fig:omega}(b) increases. 
In addition to this, we find that the users' overhead is greater than that of the BS. This is due to the fact that our algorithm optimizes the computational overhead more efficiently, making the uplink transmission, which is the responsibility of the user, the main source of overhead.
% The difference in computation and transmission is  \textcolor{blue}{because} in Algorithm \ref{algo:resourceallocation}, we optimize $\bm{f}$ first and later optimize $\bm{p},\bm{B}$ based on optimized frequencies. The frequencies $\bm{f}$ primarily affect computation time and computational energy. 

\subsection{Number of Users} Next, we explore the effect of the total number of mobile devices $N$ on optimization performance. Intuitively, when the total available resources are fixed, a larger number of $N$ could increase the overall consumption within the system as more communications and computations are required.

Specifically, we vary the number of mobile devices $N$ from $10$ to $50$, and the corresponding results are reported in Figs.~\ref{fig:number_of_users}(a) and (b).
In Fig.~\ref{fig:number_of_users}(a), we could observe an overall increasing trend in time, especially on devices. This is because the total bandwidth $B_{\textnormal{total}}$ is fixed; as the number of users $N$ increases, the average bandwidth allocated per user $n$ decreases, leading to longer transmission times. In Fig.~\ref{fig:number_of_users}(b), there is a sharper increase in energy consumption due to longer transmission times consuming more energy.
% It is worth noting that changes in computation consumption are insignificant because $N$ mainly affects the allocation of communication resources among users.

\subsection{SemCom performance Analysis}

In this section, we vary minimum PSNR value requirements $P_n^{\textnormal{min}}$ on the system and observe how the model performance and consumption will be affected. 
Intuitively, higher $P_n^{\textnormal{min}}$ leads to better performance of the SemCom but also sacrifices somewhat the optimization of overall consumption.  

We increase $P_n^{\textnormal{min}}$ from 20 dB to 30 dB to compare the PSNR values and consumption metrics. 
Specifically, Fig.~\ref{fig:Acc_Consumption_CIFAR10}(a) demonstrates that the average of PSNR values among all devices increases as $P_n^{\textnormal{min}}$ grows. The shallow blue-filled areas ``Range'' show the minimum and maximum PSNR observed in all devices.
Meanwhile, Fig.~\ref{fig:Acc_Consumption_CIFAR10}(b) illustrates that energy and time consumption also have increasing trends.
This correlation is expected: increasing $P_n^{\textnormal{min}}$ leads our algorithm to prioritize model performance for each mobile device, which is often at the expense of other optimizations like higher power usage, increased compression rate, or reduced bandwidth.

We further compare our model's performance under different training settings (i.e., performance requirement $P_n^{\textnormal{min}}$ ) with that of the typical JPEG technique, and the result is as shown in Fig.~\ref{fig:example_CIFAR}. Particularly, 
Fig.~\ref{fig:example_CIFAR}(a) presents an original example image from CIFAR-10.
Figs.~\ref{fig:example_CIFAR}(b) and (c) provide visualizations of the reconstructed images by Semcom models under different training settings.
Fig.~\ref{fig:example_CIFAR}(d) is the generated image by JPEG.
For fairness, the tested SNR values are all set at 10 dB.
We can see that even under a challenging training condition (SNR=4 dB and $\rho$ = 0.1, i.e., $P_n^{\text{min}}$=25.70 dB), our model still achieves a PSNR of 25.67 dB, comparable to that of JPEG (26.13 dB).
When the training condition is improved into (SNR=13 dB and $\rho$ = 0.1, i.e., $P_n^{\text{min}}$=26.79 dB), our PSNR value achieves 26.95 dB, outperforming JPEG.
This is because JPEG traditionally compresses images based on the pixel level and cannot resist noise effectively. In contrast, our semantic communication model accounts for channel noise during training and has learned image restoration techniques to mitigate the effects of noise.
% We also test the trained model performance with a real example image from CIFAR-10.
% Figs.~\ref{fig:example_CIFAR}(a)-(d) provides visualizations of reconstructed CIFAR images by trained models under different settings and the conventional JPEG technique.
% We can see that under a higher $P_n^{\text{min}}$ (26.79 dB), the tested PSNR value could achieve 26.95 dB, outperforming that of JPEG. 
% This is because JPEG traditionally compresses images based on the pixel level and cannot resist noise effectively. In contrast, our semantic communication model accounts for channel noise during training and has learned image restoration techniques to mitigate the effects of noise.

\begin{figure}[!t]
    \centering
    \begin{subfigure}[b]{0.49\linewidth}
        \centering
        \includegraphics[width =\linewidth]{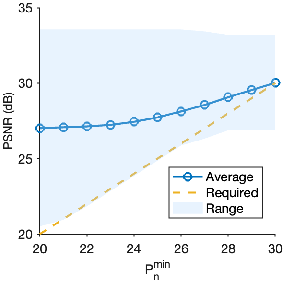}
        \caption{Overall PSNR performance.}
    \end{subfigure}
    \hfill
    \begin{subfigure}[b]{0.49\linewidth}
        \centering
        \includegraphics[width =\linewidth]{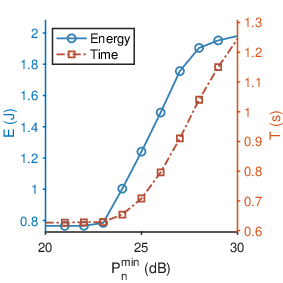}
        \caption{Total consumption.}
    \end{subfigure}
    \caption{Comparision of PSNR values and overall consumption under different $P_n^{\textnormal{min}}$}. 
    \label{fig:Acc_Consumption_CIFAR10}
\end{figure}

\begin{figure}[h]
    \centering
    \begin{subfigure}{0.47\linewidth}
        \centering
        \includegraphics[width =\linewidth]{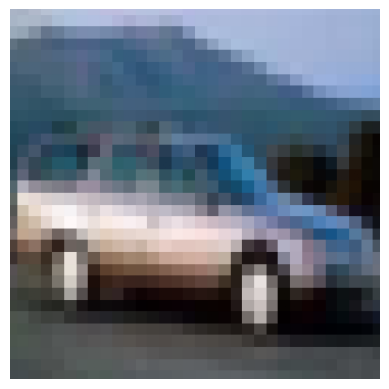}
        \caption{Original image (32$\times$32) }
    \end{subfigure}
    \hfill
    \begin{subfigure}{0.47\linewidth}
        \centering
        \includegraphics[width =\linewidth]{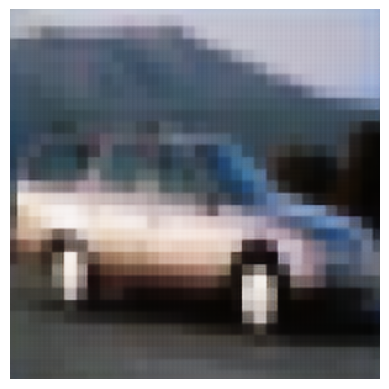}
    \caption{26.95dB(SNR=13dB, $\rho$=0.1)}
    \end{subfigure}
    \begin{subfigure}{0.47\linewidth}
       ~\\~\\ \centering
        \includegraphics[width =\linewidth]{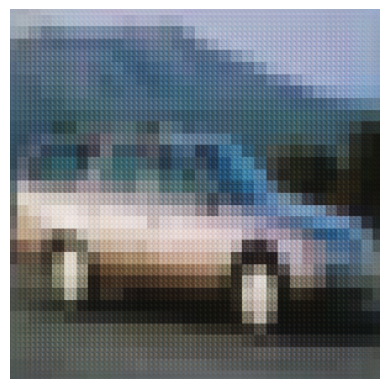}
        \caption{25.67dB(SNR=4dB, $\rho$=0.1)}
    \end{subfigure}
    \hfill
    \begin{subfigure}{0.47\linewidth}
        \centering
        \includegraphics[width =\linewidth]{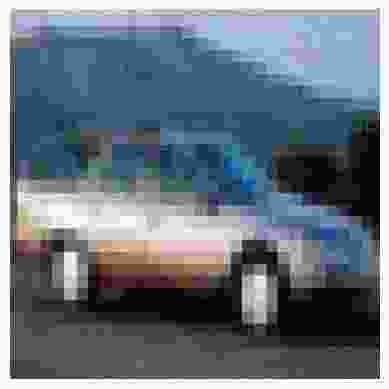}
        \caption{26.13dB (JPEG)}
    \end{subfigure}
    \caption{PSNR values of reconstructed images by our SemCom model and JPEG under different settings, e.g., (SNR=13dB, $\rho$=0.1) denotes the model trained with an SNR=13 dB and a compression rate=0.1.}
    \label{fig:example_CIFAR}
\end{figure}

\subsection{Scalability on High-Resolution Images}
Here we further test the scalability of the trained model when dealing with an image of higher resolution (512$\times$768). 
Particularly, the selected image did not appear in the training dataset (zero-shot) and is more complicated than the training ones (32$\times$32). Figs.~\ref{fig:example_highresolution}(a)-(d) shows the original image and our experimental results. From Fig.~\ref{fig:example_highresolution}(b), we can see our model could achieve a PSNR value of 25.75 dB ($\text{SNR}=13 \text{dB}$, $\rho=0.1$), only slightly lower than that of JPEG (25.89 dB) in Fig.~\ref{fig:example_highresolution}(d). When under a challenging condition (SNR=4 dB, $\rho$ = 0.1), the model could still achieve a PSNR value of 23.21 dB, as shown in Fig.~\ref{fig:example_highresolution}(c).
Thereafter, we believe that if sufficiently trained on more complex datasets, the semantic communication model can achieve performance comparable to or even beat JPEG. 

\begin{figure}
    \centering
    \begin{subfigure}{0.49\linewidth}
        \centering
        \includegraphics[width =\linewidth]{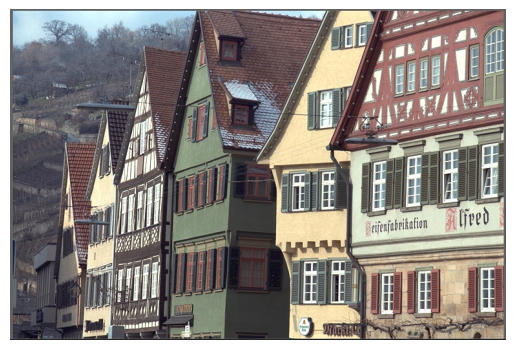}
        \caption{Original image (512$\times$768)}
    \end{subfigure}
    \begin{subfigure}{0.49\linewidth}
        \centering
        \includegraphics[width =\linewidth]{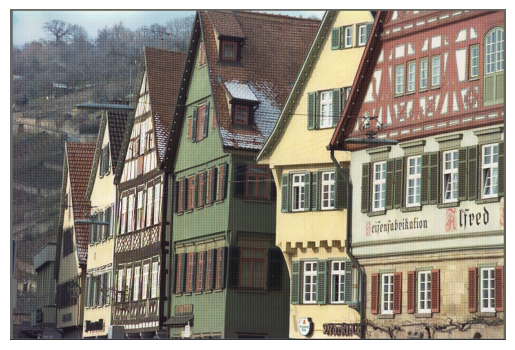}
        \caption{25.75dB (SNR=13dB, $\rho$=0.1)}
    \end{subfigure}
    ~\\ 
    \begin{subfigure}{0.49\linewidth}
        \centering
        \includegraphics[width =\linewidth]{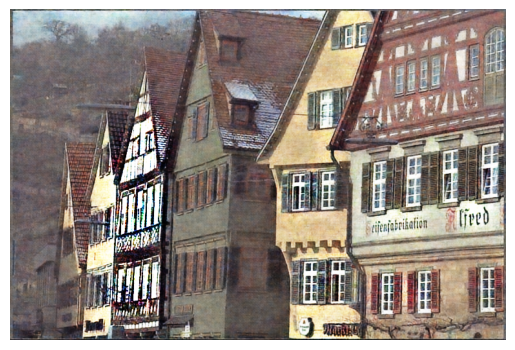}
        \caption{23.21dB (SNR=4dB, $\rho$=0.1)}
    \end{subfigure}
    \begin{subfigure}{0.49\linewidth}
        \centering
        \includegraphics[width =\linewidth]{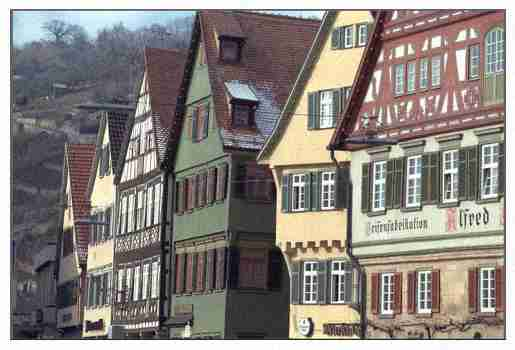}
        \caption{25.89dB (JPEG)}
    \end{subfigure}
    \caption{PSNR values of reconstructed images by our SemCom model and JPEG under different settings.}
    \label{fig:example_highresolution}
\end{figure}  

\section{Discussion \& Conclusion}\label{Sec:Conclusion}

\textbf{Generalizability:} Our current SemCom model is evaluated under AWGN channels for analytical simplicity~\cite{xie2021deep,weng2021semantic,huang2022toward,huang2022toward,zhang2023wireless}. However, it is inherently adaptable to other channel conditions. For instance, the deep JSCC model can be extended to fading channels, such as Rayleigh and Rician, by incorporating non-trainable differentiable layers, as demonstrated in~\cite{xie2021deep,bourtsoulatze2019deep,lyu2024semantic}. Additionally, our resource allocation algorithm is designed to be robust across different channel models, as it does not rely on specific channel parameters as optimization variables. This ensures the algorithm’s applicability in diverse real-world communication environments, making it suitable for various use cases.

\textbf{Deep Reinforcement Learning (DRL) Approaches.}
In this work, we rely on a centralized decision-making process and archive the ``order optimality (i.e., $\mathcal{O}(N)$)'' in time complexity for efficiency. Another promising solution to our problem is DRL.
DRL methods have demonstrated significant promise in addressing complex resource allocation challenges within distributed systems~\cite{wang2020drl,fang2022drl,khani2024deep,chua2023mobile,wenhan2024counterfactual}. These approaches leverage the ability of RL to make sequential decisions while incorporating the power of deep neural networks to process high-dimensional state and action spaces. DRL is particularly well-suited for solving non-convex optimization problems~\cite{chua2023mobile,wenhan2024counterfactual}, especially in time-sensitive scenarios.
In future work, we aim to explore DRL-based approaches to the system model, seeking the feasibility to further enhance the scalability and efficiency of resource allocation strategies. 

\textbf{Conclusion.} 
In this paper, we introduce a distributed image semantic communication system where the local devices and base station train the SemCom model collaboratively. We also study how to optimize transmission power, bandwidth, computing frequency and compression rate to minimize a weighted sum of time and energy consumption while guaranteeing the model performance. A resource allocation algorithm is further proposed to solve the joint optimization problem. Time complexity, solution quality and convergence are also analyzed. In experiments, we first compare our method with benchmarks under different resource conditions and then investigate the influence of weight parameters, number of users, and accuracy requirements. The experimental results demonstrate the effectiveness and superiority of our method.

\bibliographystyle{IEEEtran}
\bibliography{ref}

\begin{IEEEbiography}
[{\includegraphics[width=1in,height=1.25in,clip,keepaspectratio]{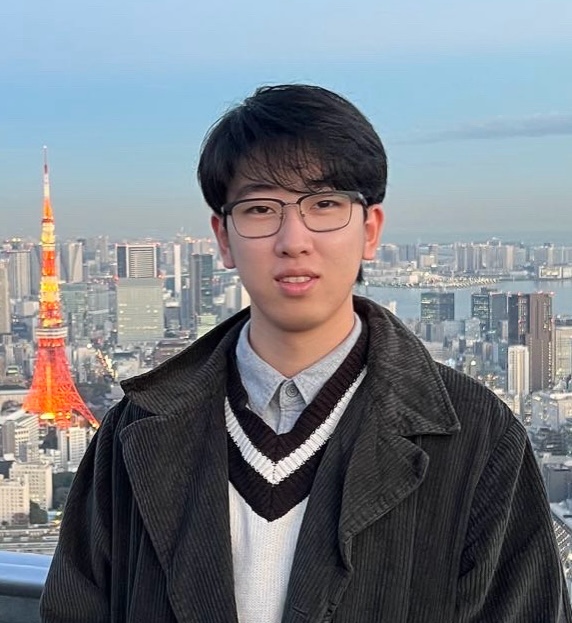}}]
{Yang Li} (Student Member, IEEE) received his B.Eng in July 2022 from Beijing University of Posts and Telecommunications (BUPT) in China. He is currently pursuing a PhD degree in the College of Computing and Data Science (CCDS) at Nanyang Technological University (NTU), Singapore. His research interest focuses on efficiency and security in AI-driven networks, as well as privacy-preserving AI technologies.
\end{IEEEbiography}

\vskip -2\baselineskip plus -1fil

\begin{IEEEbiography}
[{\includegraphics[width=1in,height=1.25in,clip,keepaspectratio]{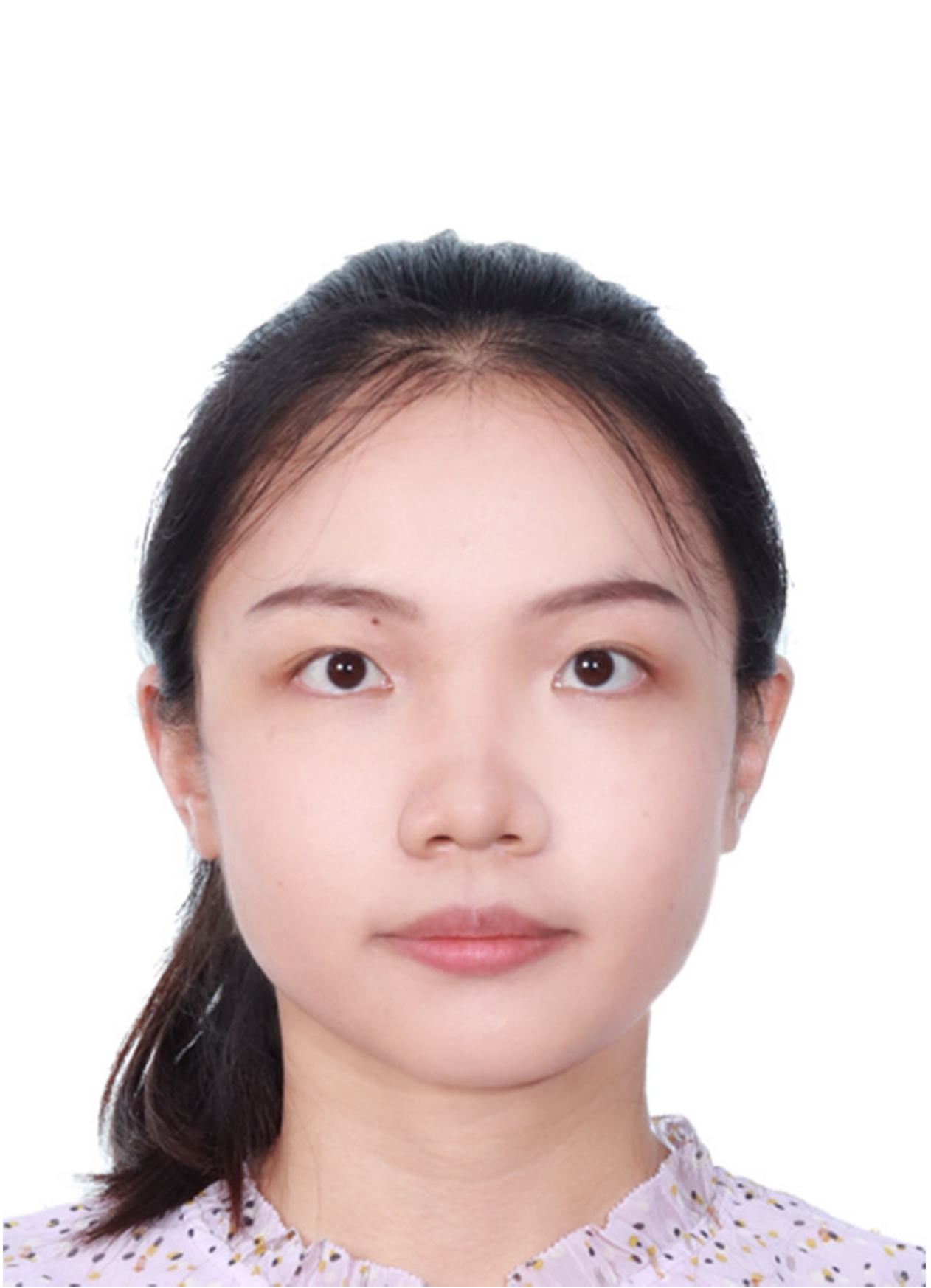}}]
{Xinyu Zhou} received her B.Eng at the department of Computer Science and Engineering from Southern University of Science and Technology (SUSTech) in 2020. She is currently a PhD student from Nanyang Technological University, Singapore. Her research interests include Optimization, Federated Learning, Wireless Communications.
\end{IEEEbiography}

\vskip -2\baselineskip plus -1fil

\begin{IEEEbiography}
[{\includegraphics[width=1in,height=1.25in,clip,keepaspectratio]{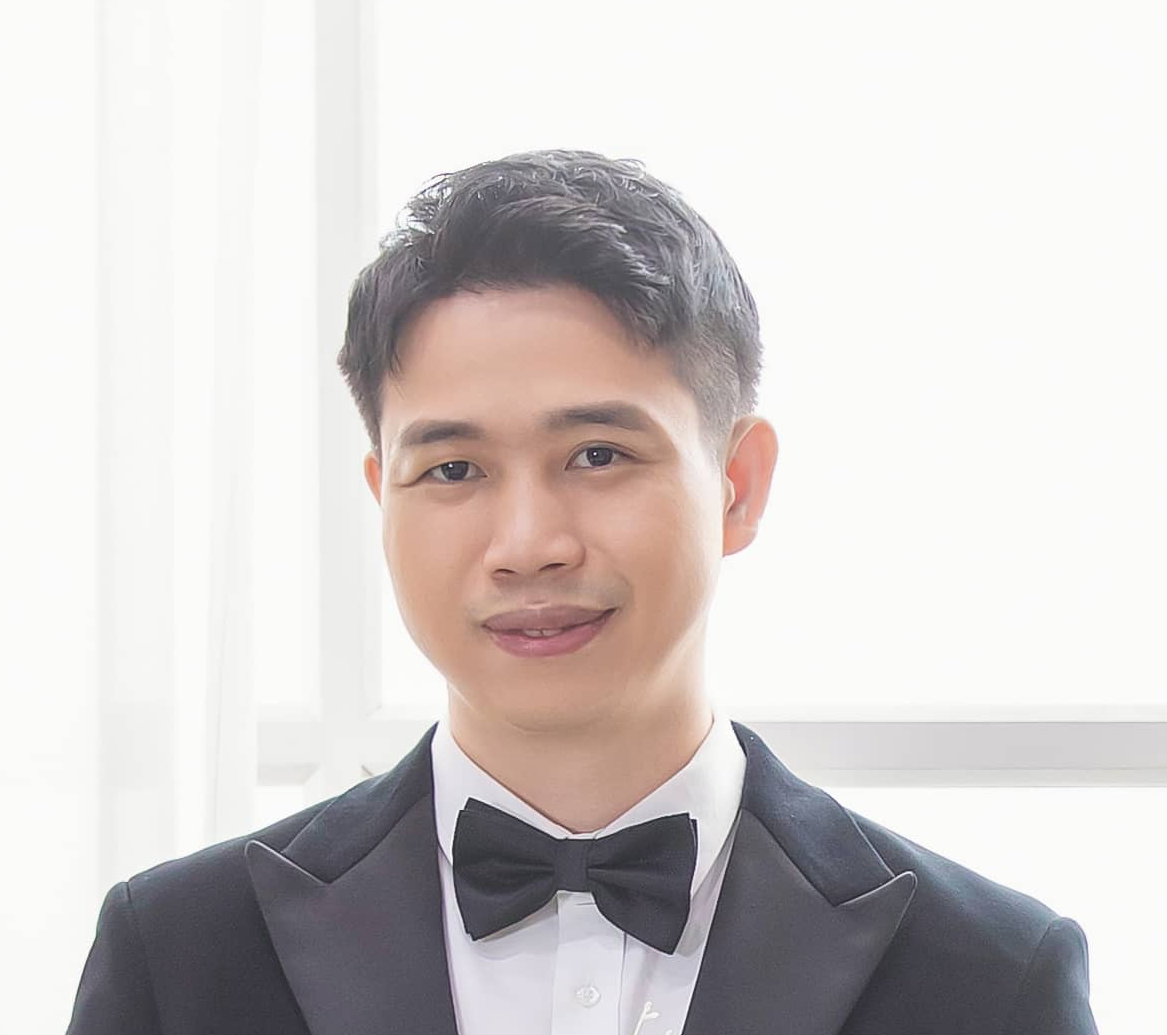}}]
{Jun Zhao} (S'10-M'15) is currently an Assistant Professor in the College of Computing and Data Science (CCDS) at Nanyang Technological University (NTU) in Singapore. He received a PhD degree in May 2015 in Electrical and Computer Engineering from Carnegie Mellon University (CMU) in the USA, affiliating with CMU's renowned CyLab Security \& Privacy Institute, and a bachelor's degree in July 2010 from Shanghai Jiao Tong University in China. 
\end{IEEEbiography}
\end{document}